\documentclass{article}
\usepackage[utf8]{inputenc}

\usepackage[preprint]{neurips_2026}
\usepackage{natbib}
\setcitestyle{numbers,square,comma}

\usepackage[utf8]{inputenc} % allow utf-8 input
\usepackage[T1]{fontenc}    % use 8-bit T1 fonts
\usepackage{hyperref}       % hyperlinks
\usepackage{url}            % simple URL typesetting
\usepackage{booktabs}       % professional-quality tables
\usepackage{amsfonts}       % blackboard math symbols
\usepackage{nicefrac}       % compact symbols for 1/2, etc.
\usepackage{microtype}      % microtypography
\usepackage{xcolor}         % colors
\usepackage{algorithm}
\usepackage{algorithmic}
\usepackage{pgfplots}
\usepackage{amsmath}
\usepackage{amsthm}
\usepackage{amssymb}
\usepackage{wrapfig}
\usepackage{tikz}
\usepackage{caption}  % 提供 \captionof 命令
\usepackage{booktabs}
\usepackage[table]{xcolor}   % 注意 [table] 选项，启用 \rowcolor
\usepackage{subcaption}
\usepackage{array}
\usepackage{graphicx}   % \resizebox 来自这里，一般已经有了

\newtheorem{theorem}{Theorem}
\newtheorem{lemma}{Lemma}
\newtheorem{assumption}{Assumption}

\definecolor{tableheader}{RGB}{42, 64, 102}
\definecolor{stripe}{RGB}{240, 244, 250}
\definecolor{sectionrow}{RGB}{222, 230, 242}
\definecolor{besthighlight}{RGB}{255, 240, 200}
\definecolor{tableheader}{RGB}{42, 64, 102}     % 深蓝表头
\definecolor{stripe}{RGB}{240, 244, 250}        % 浅蓝条纹
\definecolor{sectionrow}{RGB}{222, 230, 242}    % 分组标题底色
\definecolor{besthighlight}{RGB}{255, 240, 200} % 最优结果高亮（淡黄）

\title{Hybrid{SB}-{M}o{E}: Dual-Domain Schr\"{o}dinger Bridges with Scene-Adaptive Expert Routing for Speech Enhancement}

\author{%
  Zhengyi Lu \quad Aswini Sivakumar \quad
  Jie Hu \quad Yao Qiang \\[2pt]
  \normalfont Department of Computer Science and Engineering,
  Oakland University\\
  Rochester, MI 48309 \\
  \texttt{\{zhengyilu, aswinisivakumar, jiehu, qiang\}@oakland.edu}
}
\begin{document}

\maketitle

\begin{abstract}

Generative speech enhancement faces three gaps: spectral models capture harmonic structure but often disrupt phase, waveform models preserve phase but miss harmonics, and Schr\"odinger Bridges (SB) shorten transport from noise to clean speech but leave inference cost only loosely tied to training. We propose \textbf{HybridSB-MoE}, a dual-domain framework that fills these gaps through three contributions unified by a single asymmetric design principle.  (i) \emph{Asymmetric uncertainty fusion}: The spectral path captures epistemic uncertainty via expert disagreement, while the waveform bridge models aleatoric variance through stochastic dynamics. We fuse them asymmetrically, allowing the mixing weight to adapt to distinct error regimes rather than average predictions. (ii) \emph{Heterogeneous MoE} with top-$k{=}2$ routing across five distinct architectural archetypes, where architectural diversity makes the epistemic signal indicate which inductive bias fails rather than small perturbations among similar experts. (iii) \emph{Discretization bound} (Theorem~\ref{thm:bridge-bound}): path-consistency and trajectory regularizers together bound the $K$-step bridge sampling error in 2-Wasserstein distance at rate $K^{-\alpha}$, making small-$K$ inference an objective-level guarantee rather than an empirical claim. On VoiceBank+DEMAND, HybridSB-MoE outperforms diffusion- and SB-based baselines at their step budgets while remaining competitive with consistency-distilled few-step methods.

\end{abstract}

\section{Introduction}

Speech enhancement (SE) has long faced a structural domain trade-off. Spectral methods are effective at modeling harmonic structure and stationary noise, but often fragment phase across STFT frames; waveform methods preserve phase continuity, but can struggle with harmonically structured interference (Fig.~\ref{fig:dual_domain_concept}). This trade-off becomes especially limiting in real-world deployment, where a hearing aid on a moving bus, a video-conferencing system in a crowded caf\'e, or an in-car voice assistant must handle harmonic engine drone, broadband ambient hum, and non-stationary crowd babble under millisecond latency budgets and without scene labels at inference time \cite{wang2018overview,defossez2020demucs}. These settings require more than a universally stronger model: they require a mechanism that can exploit complementary inductive biases and decide which one to trust for each input.

A decade of deep learning has advanced SE from DNN masking \cite{xu2015reg,pascual2017segan}, through time-domain \cite{luo2019convtasnet,defossez2020demucs,luo2020dual} and spectro-temporal \cite{chen2022fullsubnetplus,wang2023tfgridnet} architectures to diffusion-based generative models \cite{ho2020ddpm,welker2022speech,richter2023sgmseplus} and Schr\"odinger Bridge formulations that replace the uninformative Gaussian prior of diffusion with direct transport from noisy to clean speech. Despite this progress, three structural limitations persist that are precisely what a dual-domain framework must address. \textbf{(i) Single-domain commitment.} Existing methods typically operate in either the waveform or spectral domain, sacrificing the complementary inductive bias of the other \cite{wang2023tfgridnet}. \textbf{(ii) Uniform processing across heterogeneous noise.} A single network handles stationary appliance hum, harmonic engine noise, and non-stationary crowd babble alike, despite their structurally different inductive biases \cite{sivaraman2020smle,chazan2021mode}. \textbf{(iii) Loosely controlled sampling cost.} Generative SE pipelines often require many iterative refinement steps, with limited formal connection between the training objective and the inference budget \cite{welker2022speech,richter2023sgmseplus}. 

The natural response, i.e., running a spectral and a waveform pathway in parallel and combining them, fails on its own, because a generic ensemble of two prediction streams provides no mechanism to know which pathway is failing on a given input. Without such signal, the combination defaults to fixed-weight averaging, which is provably suboptimal whenever the two branches' errors are uncorrelated, and the dual-domain promise reduces to a parameter-count increase. 

\begin{wrapfigure}{r}{0.42\textwidth}
  \vspace{-8pt}
  \centering
  \includegraphics[width=0.40\textwidth]{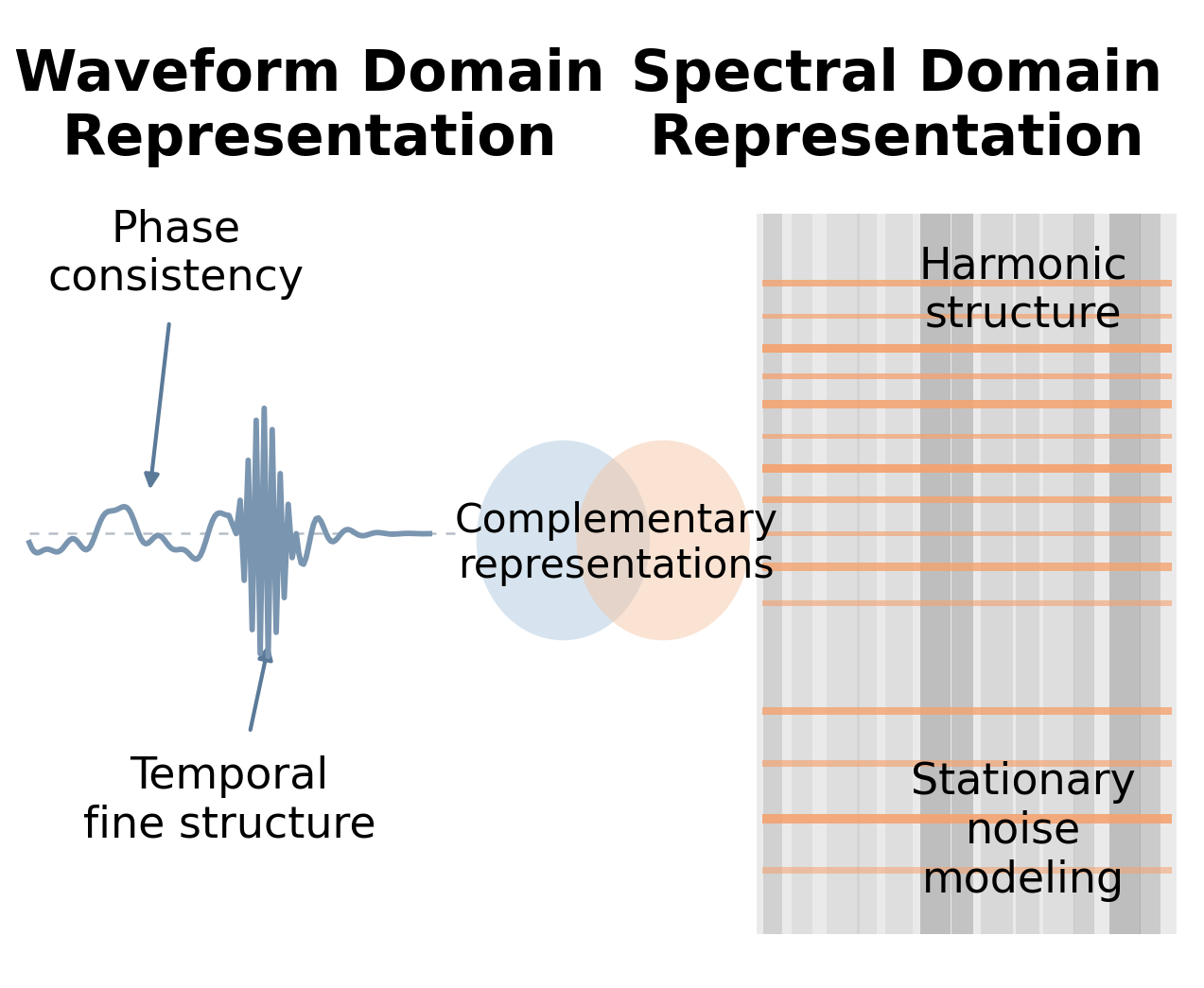}
  \vspace{-4pt}
  \caption{The persistent dual-domain dichotomy in speech enhancement. Waveform processing preserves temporal fine structure and phase coherence but underperforms on harmonically structured interference; spectral processing captures harmonic structure and stationary noise patterns but fragments phase across STFT frames. The two domains succeed and fail in opposite regimes, and existing generative methods inherit one half of the trade-off by committing to a single domain.}
  \label{fig:dual_domain_concept}
  \vspace{-8pt}
\end{wrapfigure}

We introduce \textbf{HybridSB-MoE}, a dual-domain framework organized around one asymmetric idea: \emph{a multi-expert spectral path naturally produces an epistemic uncertainty signal (which expert is right?), and a stochastic waveform bridge naturally produces an aleatoric one (intrinsic transport noise). These are categorically different signals about categorically different errors}. 
Pairing the two pathways and fusing on this asymmetry yields a mixing weight that selects between two error regimes rather than averaging two predictions, transforming the dual-domain promise into a principled co-design. Two further architectural choices are required for this idea to work. 
First, the spectral pathway must produce \emph{informative} expert disagreement, which requires architectural heterogeneity rather than capacity-replicated variation. We therefore use heterogeneous expert archetypes, each aligned with a canonical noise-processing primitive, including low-rank denoising, wide receptive fields, information bottlenecks, harmonic bases, and universal approximation, and combine them under sparse routing. Second, the waveform pathway must remain faithful under the small inference budget $K$ targeted by the system. We therefore train it with path-consistency and trajectory regularizers, and prove in Theorem~\ref{thm:bridge-bound} that jointly minimizing them bounds the $K$-step bridge discretization error in $2$-Wasserstein distance, making small-$K$ inference a consequence of the training objective rather than an empirical heuristic. In Section \ref{sec:method}, we further show that no proper subset of these three components recovers the central design property.

Our contributions are as follows:

\textbf{1. A dual-domain framework built around asymmetric uncertainty.} We propose HybridSB-MoE, a dual-domain framework that couples spectral MoE and waveform SB pathways. By pairing each path with its natural uncertainty signal, epistemic disagreement for spectral experts and aleatoric variance for the waveform bridge, HybridSB-MoE enables fusion that adapts to error regimes rather than fixed-weight averaging.

\textbf{2. A heterogeneous spectral MoE for informative epistemic disagreement.} We design a sparse top-$k{=}2$ spectral MoE with architecturally distinct expert archetypes. This heterogeneity makes expert disagreement reflect mismatched inductive biases rather than small perturbations among capacity-replicated experts.

\textbf{3. A regularized waveform bridge for few-step inference.}
We train the waveform SB with path consistency and trajectory regularizers, and prove in Theorem~\ref{thm:bridge-bound} that they bound the $K$-step discretization error in $2$-Wasserstein distance, linking small-$K$ inference to the training objective.

\textbf{4. SOTA on VoiceBank+DEMAND.} HybridSB-MoE improves over diffusion- and SB-based baselines at their respective step budgets, remains competitive with consistency-distilled few-step methods, and provides calibrated fusion uncertainty for input-adaptive dual-domain enhancement.

\section{Related Work}
\vspace{-4 mm}
\textbf{Discriminative SE.} The field has progressed from DNN masking \cite{xu2015reg,pascual2017segan} that surpassed classical Wiener and MMSE estimators \cite{boll1979spectral,ephraim1984mmse}, through time-domain encoder--decoders (Conv-TasNet \cite{luo2019convtasnet,leroux2019sisdr}, DEMUCS \cite{defossez2020demucs}), dual-path sequence models \cite{luo2020dual}, full-sub-band hybrids (FullSubNet+ \cite{chen2022fullsubnetplus}), spectro-temporal grids (TF-GridNet \cite{wang2023tfgridnet}), and recent Mamba-based linear-complexity backbones \cite{chao2024semamba,wang2024mambaseunet}. Joint enhancement objectives that combine SI-SDR with perceptual surrogates push fidelity further \cite{yen2024joint,yousif2025speech}, yet the underlying paradigm remains a one-shot deterministic regression. Two limitations therefore persist regardless of architecture. \emph{(i) Inherent ambiguity:} under severe non-stationary or harmonically structured noise, the noisy-to-clean mapping admits multiple plausible solutions, and a deterministic regressor must collapse this multimodal posterior to a single point estimate, which typically manifests as residual noise or over-suppressed harmonics \cite{chen2022likelihood}. \emph{(ii) No uncertainty:} discriminative networks expose no signal at inference time to flag this ambiguity, so a downstream module cannot tell when the prediction should be trusted---making them unsuitable as one branch of a confidence-routed dual-domain combiner.

\textbf{Generative models and Schr\"odinger Bridges for SE.} Likelihood-based \cite{chen2022likelihood} and score-based \cite{ho2020ddpm,song2021scorebased,song2020denoising} diffusion models reverse a noise corruption process; SGMSE/SGMSE+ \cite{welker2022speech,richter2023sgmseplus} apply this in the complex spectrogram domain, jointly modeling magnitude and phase, but require many iterative steps from an uninformative Gaussian prior. Two recent paradigms exploit the structure of the noisy observation directly. \emph{Flow-matching} approaches \cite{lipman2023flow,liu2022flow} regress velocity fields between coupled distributions, shortening the sampling trajectory but typically without an explicit stochastic perturbation. \emph{Schr\"odinger Bridges} \cite{debortoli2021diffusion,shi2023dsbm} learn the optimal-transport-like coupling between the noisy observation and clean speech, retaining a small bridge variance for full-support marginals. SB for SE \cite{jukic2024sbse,wang2024sbse,tang2024simpledsb,yang2025schr,ZHANG2025SBSE} confirm improved low-SNR structure preservation along this shorter path. Empirical step-count reductions have been achieved by consistency distillation \cite{xu2025rosecd,song2023consistency}, adversarial few-step training \cite{han2025fewstep}, and SB-consistency trajectory models \cite{nishigori2025schr}. None of these works links the training objective to the inference budget, where small $K$ is \emph{asserted} but not \emph{derived}. Theoretical convergence analyses for diffusion sampling exist in the broader literature \cite{chen2023sampling}, but their bounds are stated in terms of the score-matching error rather than tractable, training-time-minimized regularizers, and have not been adapted to the doubly-conditioned SB setting. Moreover, these methods remain single-domain and lack calibrated uncertainty for fusion. We address both with asymmetric uncertainty fusion and a discretization bound  (Theorem \ref{thm:bridge-bound}) that links small-$K$ inference to explicit training-time quantities.

\textbf{Mixture-of-Experts for SE.} MoE enables conditional computation with sparse routing, scaling capacity without proportional inference cost \cite{shazeer2017outrageously,fedus2022switch,lepikhin2021gshard}. In SE, sparse MoE \cite{sivaraman2020smle} first enabled scene-adaptive processing via noise-clustered experts; zero-shot personalization \cite{sivaraman2021zeroshot} extended this to speakers via embedding-based routing; clean-cluster pre-training \cite{chazan2021mode} enforced expert specialization through pre-training initialization; and dynamic capacity allocation \cite{miccini2025adaptive} improved routing efficiency at inference. However, this line of work shares three limitations. \emph{(i) Homogeneity.} Existing methods replicate a single architecture across experts, assuming one computational primitive can handle tonal, ambient, harmonic, and crowd noise. This assumption breaks down under realistic, structurally diverse noise. \emph{(ii) Routing without uncertainty.} Existing routing only partitions inputs into specialist regions, rather than exposing \emph{interpretable uncertainty}; thus, the epistemic signal from multi-expert disagreement remains unused. \emph{(iii) Separation from generative SE.} MoE routing and generative bridge models remain disconnected in SE, leaving scene-adaptive computation and structured distributional transport as separate lines of work. We address all three with heterogeneous archetype experts that encode distinct noise-specific inductive biases. Their disagreement provides the epistemic signal which, together with the SB's aleatoric variance, drives asymmetric fusion. To our knowledge, this is the first joint MoE--SB framework for SE.

\textbf{Uncertainty quantification and dual-branch fusion.} Calibrated uncertainty for deep regression is most commonly obtained via deep ensembles \cite{lakshminarayanan2017simple} or learned heteroscedastic variance heads, with classification calibration assessed via ECE \cite{guo2017calibration}. Prior dual-branch SE methods that combine time- and frequency-domain estimates typically use either fixed-weight averaging or learned attention over predictions \cite{wang2023tfgridnet,defossez2020demucs}; none assigns categorically distinct uncertainty types to the two branches, and so the fusion cannot route on \emph{which kind} of error each branch is exposed to. Our pathway-typed asymmetric fusion is, to our knowledge, the first to exploit this structure, and the calibration loss (Eq. (\ref{eq:calibration})) ties the abstract epistemic/aleatoric labels to measurable per-pathway reconstruction errors so the routing weight is learned, not asserted.

%===========================================================
\section{Method}
\label{sec:method}
\vspace{-2 mm}
HybridSB-MoE has three coupled components: a heterogeneous spectral MoE, an SB waveform pathway with path-consistency and trajectory regularizers (for which we prove a $K$-step discretization bound, Theorem \ref{thm:bridge-bound}), and an asymmetric uncertainty fusion that selects between epistemic and aleatoric error regimes. Figure~\ref{fig:architecture} illustrates the overall HybridSB-MoE framework.  We develop each component in turn and close the section by showing that no proper subset recovers the central design property.

\begin{figure*}[t]
    \centering
    \includegraphics[width=0.98\textwidth]{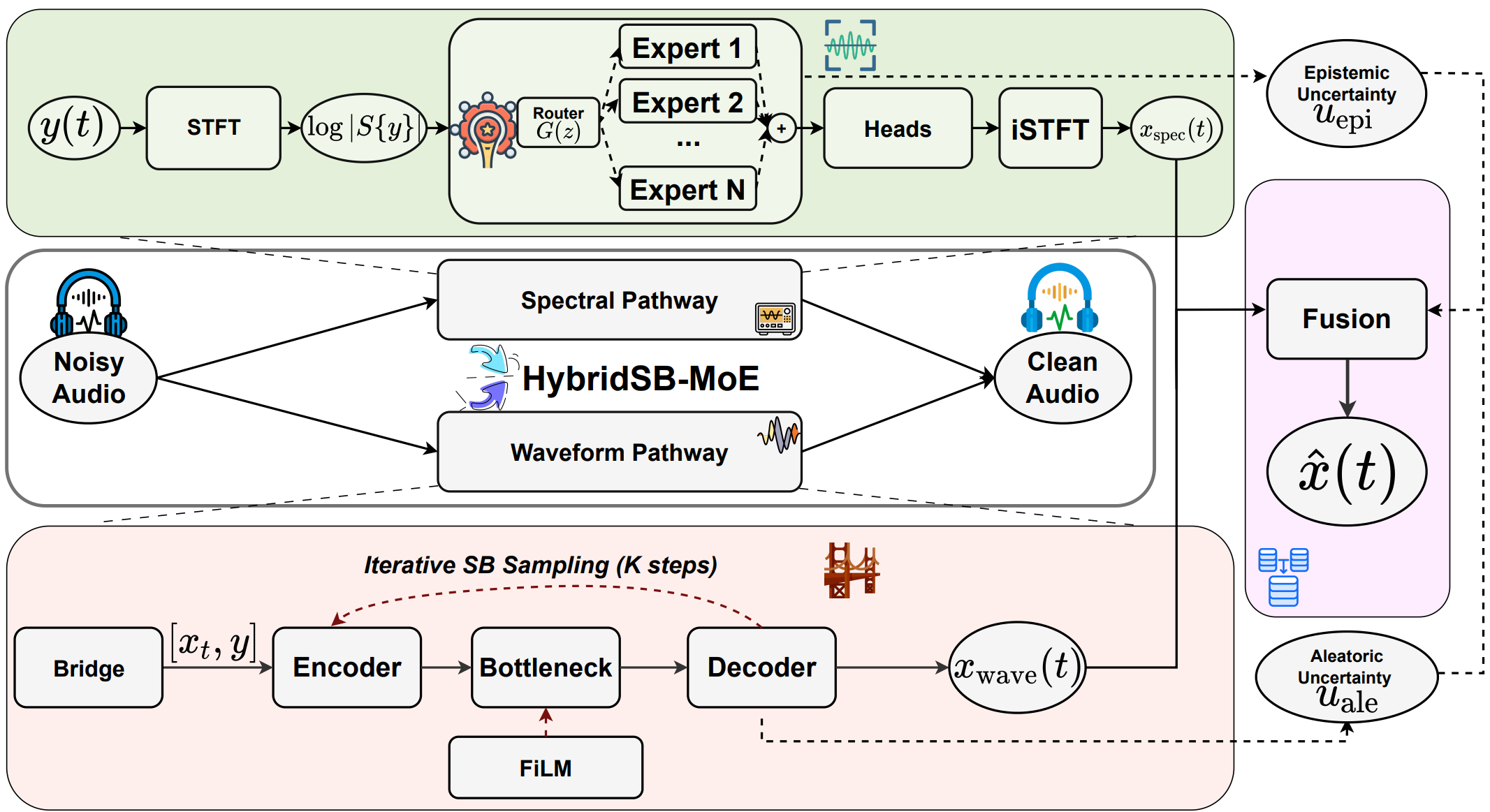}
    \vspace{+2 mm}
    \caption{Overview of HybridSB-MoE. The spectral pathway (top) routes log-magnitude features $z = \log|S\{y\}|$ through $N$ heterogeneous archetype experts via gating $G(z)$ to produce $x_{\text{spec}}(t)$. The waveform pathway (bottom) iteratively refines an SB state $[x_t, y]$ through a U-Net to produce $x_{\text{wave}}(t)$. An uncertainty-aware fusion combines both via $u_{\text{epi}}$ (top-$k$ expert disagreement) and $u_{\text{ale}}$ (bridge-variance head) to yield $\hat{x}(t)$.}
    \vspace{-5 mm}
    \label{fig:architecture}
\end{figure*}

\textbf{Overview.} Given a noisy input $y(t)$, HybridSB-MoE processes it through two parallel pathways designed to capture complementary speech characteristics. The spectral pathway applies the heterogeneous MoE to the log-magnitude STFT for scene-adaptive enhancement. In parallel, the waveform pathway runs an efficient SB on the time-domain signal to preserve phase. The two pathways exhibit asymmetric reliability across acoustic conditions: spectral processing excels in harmonically structured noise, whereas waveform processing is more robust in phase-sensitive scenarios. Consequently, a fixed fusion strategy is inherently suboptimal. We therefore introduce an uncertainty-aware fusion module that integrates $x_\mathrm{spec}(t)$ and $x_\mathrm{wave}(t)$ by dynamically weighting each pathway based on confidence, producing the final enhanced signal $\hat{x}(t)$. Both pathways run in parallel at inference, and long utterances use overlap-add segmentation (Appendix \ref{app:training}).

\textbf{Spectral feature extraction.} We apply the STFT \cite{Allen1977STFT} to obtain the complex spectrogram $S\{y\}\in\mathbb{C}^{F\times T_f}$ ($F$ frequency bins) and decompose it in polar form into magnitude $|S\{y\}|$ and phase $\angle S\{y\}$, processed by separate heads (below). The MoE consumes log-magnitude features $z=\log|S\{y\}|\in\mathbb{R}^{F\times T_f}$, whose logarithmic scaling compresses the dynamic range in line with human hearing.

\textbf{Heterogeneous mixture-of-experts.}\label{sec:moe} Unlike conventional homogeneous MoE designs \cite{shazeer2017outrageously,fedus2022switch,lepikhin2021gshard} in which all experts share the same architecture and differ only in learned weights, we instantiate distinct architectural archetypes, each tailored to a canonical noise-processing primitive, i.e., \textit{Home, Nature, Office, Transport,} and \textit{Public} for VoiceBank+DEMAND (see Appendix \ref{app:experts} for full specification). Under sparse top-$k$ routing, pairwise mixtures and weighted blends of these archetypes cover the $14$ noise categories in VoiceBank+DEMAND \cite{thiemann2013demand} \emph{combinatorially} rather than by a dedicated expert per noise, and the design extends to new acoustic environments by adding archetypes per noise family encountered, without altering the rest of the framework (see Appendix \ref{app:experts} for details). All experts share a backbone of normalization (group norm \cite{wu2018group}) $\to$ variable internal modules $\to$ bottleneck projection, ensuring routing-interface compatibility while permitting structurally divergent intermediate computation. This design, a small semantically anchored basis composed of sparse routing, replaces the parameter scaling of large homogeneous MoEs with interpretable routing and meaningful expert disagreement, which we use for asymmetric fusion below.

\textbf{Scene-adaptive sparse routing.}\label{sec:routing} Standard sparse routing \cite{shazeer2017outrageously,fedus2022switch} assigns each input token to its top-$k$ experts independently, which is well-suited to language modelling but ill-suited to SE: the dominant noise type is a global property of an utterance, while local acoustic events (a phoneme onset, a transient hit) require frame-level adaptivity. We therefore use a two-level gate: \emph{Archetype-level} routing $G_{\mathrm{arch}}(z)$ pools $z$ across time and decides the dominant archetype mixture for the whole utterance; \emph{token-level} routing $G_{\mathrm{token}}(z)$ refines this assignment per frame. The two combine as $G(z)=\alpha G_{\mathrm{arch}}(z)+(1-\alpha)G_{\mathrm{token}}(z)$ with $\alpha\in[0,1]$. An MLP maps $G(z)$ to per-expert logits; the top-$k$ entries are softmax-renormalized to weights $\{G_i(z)\}_{i\in\mathcal{I}_k}$, and the routed output is
\begin{equation}
\hat{x}_{\text{spec}} = \sum_{i \in \mathcal{I}_k} G_i(z)\, E_i(z),
\label{eq:moe_output}
\end{equation}

where $\mathcal{I}_k$ indexes the selected experts and $E_i$ are expert networks. To prevent expert collapse \cite{shazeer2017outrageously,fedus2022switch}, a standard load-balancing regularizer
\begin{equation}
\mathcal{L}_{\mathrm{aux}} = \lambda_I\,\mathrm{Var}(p_i) + \lambda_L\,\mathrm{Var}(n_i)
\label{eq:load_balance}
\end{equation}
penalizes uneven utilization, where $p_i$ and $n_i$ denote the average routing probability and token count of expert $i$, respectively.

\textbf{Spectral reconstruction.}\label{sec:spec_recon} The routed features $\hat{x}_{\mathrm{spec}}$ pass through two parallel $1\times1$-conv heads. Since phase cannot be inferred from log-magnitude alone, the phase head also consumes the noisy phase as $(\sin\angle S\{y\},\cos\angle S\{y\})$, concatenated along the channel dimension:
\begin{equation}
\hat{M} = M_{\max}\,\mathrm{sigmoid}(h_{\mathrm{mag}}(\hat{x}_{\mathrm{spec}})), \Delta\phi = \phi_{\max}\,\tanh(h_{\mathrm{pha}}(\hat{x}_{\mathrm{spec}}, \sin\angle S\{y\}, \cos\angle S\{y\})),
\label{eq:moe_heads}
\end{equation}
where $M_{\max},\phi_{\max}$ bound the mask and phase correction to prevent over-suppression, excessive amplification, and phase-wrap instability. The enhanced spectrum $\hat{S}=\hat{M}\odot|S\{y\}|\cdot e^{j(\angle S\{y\}+\Delta\phi)}$ is inverted to $x_{\mathrm{spec}}(t)$ via iSTFT with overlap-add. STFT processing nevertheless fragments temporal fine structure through frame-wise phase decoupling, motivating the parallel waveform pathway below.

\textbf{Schr\"odinger Bridge formulation.} The continuous-time SB problem seeks the stochastic process minimizing the KL divergence to a reference Brownian motion subject to fixed marginal distributions at both endpoints \cite{debortoli2021diffusion,shi2023dsbm,peyre2019computational}. In our application the two endpoints are the noisy observation $y$ and the clean target $x$, so the bridge directly couples the conditional distributions our system must transport between. Where standard diffusion \cite{ho2020ddpm} traverses the full path from a data-independent Gaussian prior, the bridge exploits the substantial mutual information between $y$ and $x$ to yield a shorter, more informative trajectory and, empirically, fewer sampling steps and better low-SNR structure preservation \cite{jukic2024sbse,wang2024sbse,tang2024simpledsb}. Conceptually adjacent flow-matching methods \cite{lipman2023flow,liu2022flow} also couple paired distributions but typically without an explicit bridge perturbation; our SB construction retains a small noise term $\sigma_t\epsilon$ to keep the marginal $p(x_t\mid x,y)$ full-support, which is required for score-based reverse sampling and for the trajectory regularizer to act as a self-consistency constraint.

The intermediate state at diffusion time $t \in [0,T]$ interpolates between $y$ and $x$ with a small Gaussian perturbation,
\begin{equation}
x_t = \sqrt{\bar{\beta}_t}\, x + \sqrt{1 - \bar{\beta}_t}\, y + \sigma_t\, \epsilon, \qquad \epsilon \sim \mathcal{N}(0, I),
\label{eq:sb_forward}
\end{equation}
where $\bar{\beta}_t = f(t)/f(0)$ with $f(t)=\cos^2((t/T+s)/(1+s)\cdot\pi/2)$ is a cumulative cosine schedule \cite{nichol2021improved} decreasing monotonically from $\bar{\beta}_0=1$ to $\bar{\beta}_T=0$ (offset $s{=}0.008$), and $\sigma_t = \sigma_{\max}\sqrt{\bar{\beta}_t(1-\bar{\beta}_t)}$ ($\sigma_{\max}{=}0.05$) vanishes at both endpoints, so $x_0 = x$ and $x_T = y$ deterministically while $\sigma_t\epsilon$ supplies the full-support marginal $p(x_t\mid x,y)$ required by the score-based reverse process.

The reverse process is parameterized as a learnable denoiser $\hat{x}_\theta(x_t, y, t)$ that predicts the clean target directly from the bridge state and the observation. Rather than solving the reverse-time stochastic differential equation (SDE) associated with the bridge in closed form, we adopt a tractable bridge-consistent update that iteratively re-injects the predicted clean estimate into the forward construction of Eq. (\ref{eq:sb_forward}) at the previous timestep:
\begin{equation}
x_{t-1} = \sqrt{\bar{\beta}_{t-1}}\, \hat{x}_\theta(x_t, y, t) + \sqrt{1 - \bar{\beta}_{t-1}}\, y + \sigma_{t-1}\, z, \qquad z \sim \mathcal{N}(0, I).
\label{eq:sb_sampling}
\end{equation}
This data-prediction update is the bridge analogue of DDPM/DDIM ancestral sampling \cite{ho2020ddpm,song2020denoising}: deterministic when $\sigma_{t-1}{=}0$, stochastic otherwise. Bridge-marginal consistency is enforced empirically through the path-consistency and trajectory regularizers introduced below. The denoiser is a 1D U-Net with transformer bottleneck and FiLM timestep conditioning (Appendix \ref{app:training}). The training objective is the data-prediction loss:
\vspace{-2mm}
\begin{equation}
\mathcal{L}_{\mathrm{SB}}^{\mathrm{data}} = \mathbb{E}_{t,(x,y),\epsilon} \big\| \hat{x}_\theta(x_t, y, t) - x \big\|_2^2,
\label{eq:sb_data_loss}
\end{equation}
augmented by the path-consistency and trajectory regularizers introduced next.

\textbf{Path-consistency and trajectory regularizers.}\label{sec:efficient_sampling} A fixed front-loaded schedule $t_k = T(k/K)^\gamma$ with $\gamma{<}1$ concentrates the $K$-step budget in the early-trajectory regime where the SNR changes fastest and where errors propagate to all later steps \cite{ho2020ddpm}. Two regularizers then encourage a smooth, low-curvature denoising path in the $(y,x)$ subspace. The path-consistency loss enforces cross-timestep agreement of clean-signal predictions along the same trajectory:
\begin{equation}
\mathcal{L}_{\mathrm{path}} = \mathbb{E}_{t,t'}\left[\left\| \hat{x}_\theta(x_t, y, t) - \hat{x}_\theta(x_{t'}, y, t') \right\|^2\right]
\label{eq:path_loss}
\end{equation}
where $t,t'$ are independently sampled along the same bridge trajectory. This penalizes inconsistent clean-signal predictions across timesteps, reducing reverse-process curvature (cf.\ consistency models \cite{song2023consistency}, but here for path smoothing, not distillation). The trajectory loss anchors $x_t$ to the schedule-consistent reconstruction:
\begin{equation}
\mathcal{L}_{\mathrm{traj}} = \mathbb{E}_t\left[\left\| x_t - \big(\sqrt{\bar{\beta}_t}\, \hat{x}_\theta(x_t, y, t) + \sqrt{1-\bar{\beta}_t}\, y\big) \right\|^2\right].
\label{eq:traj_loss}
\end{equation}
At the optimum $\hat{x}_\theta=x$, the residual reduces to $\sigma_t\epsilon$, so $\mathcal{L}_{\mathrm{traj}}$ couples the prediction to the specific bridge state $x_t$ at each timestep, complementing the in-expectation regression of Eq. (\ref{eq:sb_data_loss}). At inference, we use $K{=}8$ steps with the front-loaded non-uniform discretization $t_k = T(k/K)^\gamma$, $\gamma{=}0.6$ (see Appendix \ref{app:training} for more details).

\textbf{Discretization bound.}\label{sec:bound} The two regularizers above make small-$K$ inference viable, and we now show \emph{why}. Let $\hat p_K$ be the law of the $K$-step rollout of Eq. (\ref{eq:sb_sampling}) at $t{=}0$ and $p_0^{\mathrm{br}}$ the continuous-time bridge marginal at $t{=}0$, both conditional on $y$.

\begin{assumption}[Regularity]
\label{assu:reg}
(i) $\hat{x}_\theta(\cdot, y, t)$ is $L_x$-Lipschitz in its first argument and $L_t$-Lipschitz in $t$. (ii) $\bar\beta_t$ and $\sigma_t = \sigma_{\max}\sqrt{\bar\beta_t(1-\bar\beta_t)}$ are $C^1$ on $[0,T]$, with $\bar\beta_t$ monotone and bounded second derivative. (iii) $\mathbb{E}\|x_t\|^2$ and $\mathbb{E}\|\hat x_\theta(x_t,y,t)\|^2$ are uniformly bounded in $t$.
\end{assumption}

\begin{theorem}[K-step bridge discretization bound]
\label{thm:bridge-bound}
Under Assumption \ref{assu:reg}, with the schedule $t_k = T(k/K)^\gamma$, $\gamma\in(0,1]$,
\begin{equation}
\label{eq:bound}
W_2\!\left(\hat p_K,\, p_0^{\mathrm{br}}\right) \;\leq\; C_1\, K^{-\alpha} \;+\; C_2\, \sqrt{\mathcal{L}^\star_{\mathrm{path}} + \mathcal{L}^\star_{\mathrm{traj}}},
\end{equation}
with $\alpha=\min(1,\gamma)$, $\mathcal{L}^\star$ the training-termination values of Eq. (\ref{eq:path_loss})–(\ref{eq:traj_loss}), and $C_1, C_2$ depending only on $L_x,L_t,\sigma_{\max}$, and the curvature of $\bar\beta_t$.
\end{theorem}
With $\alpha=\min(1,\gamma)$, front-loading ($\gamma<1$) concentrates steps near $t{=}0$ where $\Phi'$ is largest, which is what drives small-$K$ behavior empirically. Each term corresponds to a separate lever: $C_1 K^{-\alpha}$ is controlled by inference-time $K,\gamma$; $C_2\sqrt{\mathcal{L}^\star}$ is controlled by training-time minimization of the two regularizers and is independent of $K$. Two consequences follow. First, once $\mathcal{L}^\star$ is small, the bound saturates: any $K$ with $C_1 K^{-\alpha}\!\lesssim\! C_2\sqrt{\mathcal{L}^\star}$ already attains the asymptotic floor, which is what justifies our $K{=}8$ rather than the much larger $K$ of standard diffusion sampling. Second, the regularizers are not auxiliary terms; without them, the second term has no training-time upper bound, and the inequality does not exist as a function of $K$. Therefore, removing either breaks the entire argument. This is the formal sense in which path-consistency and trajectory regularization are load-bearing for our inference schedule. Eq. (\ref{eq:bound}) is closely related to recent convergence results for diffusion sampling \cite{chen2023sampling} but adapts those analyses to the doubly-conditioned bridge process and, importantly, replaces score-matching error with the explicit, training-time-minimized regularizer pair, which makes the bound directly actionable as a design tool. The full proof includes a one-step fidelity lemma combined with a non-uniform Riemann estimate under a synchronous coupling, with details deferred to Appendix \ref{app:proof}. We present Eq. (\ref{eq:bound}) as a design-justifying inequality, not a tight predictor.

\textbf{Asymmetric uncertainty fusion.}\label{sec:fusion} A naive dual-branch ensembler treats the two pathways symmetrically; ours does not, by design. Epistemic uncertainty (`which expert is right?') is well-defined only when multiple experts coexist (the spectral MoE); the waveform pathway has no analog. Aleatoric uncertainty (intrinsic stochasticity of the bridge SDE in Eq. (\ref{eq:sb_forward})) is well-defined only when the forward process is genuinely stochastic (the waveform pathway); the deterministic spectral mask in Eq. (\ref{eq:moe_heads}) has no analog. Each pathway therefore carries exactly one uncertainty type: $u_{\mathrm{epi}}$ flags an architectural mismatch, $u_{\mathrm{ale}}$ flags a stochastic transport error. The fusion selects between these two error regimes; it does not average two predictions. Concretely, $u_{\mathrm{epi}} \!=\! \tfrac{1}{kT_f}\sum_{i\in\mathcal{I}_k}\|E_i(z)\!-\!\bar E(z)\|_2^2$ with $\bar E(z)\!=\!\tfrac{1}{k}\sum_{i\in\mathcal{I}_k}E_i(z)$ (deep-ensembles sense \cite{lakshminarayanan2017simple}, not Bayesian posterior); $u_{\mathrm{ale}}$ is from the U-Net's per-sample log-variance head, exponentiated and time-averaged. After z-score normalization, a 2-layer MLP gives $w\!=\!\sigma(\mathrm{MLP}(\tilde u_{\mathrm{epi}},\tilde u_{\mathrm{ale}}))\!\in\![0,1]$, and
\begin{equation}
\hat{x}(t) = w \cdot x_{\mathrm{spec}}(t) + (1-w)\cdot x_{\mathrm{wave}}(t).
\label{eq:hatx}
\end{equation}
High expert disagreement pushes $w$ toward the waveform pathway and high SB variance toward the spectral pathway: each pathway is deferred away from when its native error is large. A calibration loss anchors the two scalars to the errors they should track,
\begin{equation}
\mathcal{L}_{\mathrm{cal}} = (u_{\mathrm{epi}} - \|x_{\mathrm{spec}}-x\|_2^2)^2 + (u_{\mathrm{ale}} - \|x_{\mathrm{wave}}-x\|_2^2)^2,
\label{eq:calibration}
\end{equation}
where $u_{\mathrm{epi}},u_{\mathrm{ale}}$ are pre-normalized; without it, the scalars have no incentive to scale with error, and the fusion's dependence on them would be arbitrary.

\textbf{Training.}\label{sec:training} All components are trained end-to-end. Each step runs both pathways in parallel: the spectral path produces $x_{\mathrm{spec}},u_{\mathrm{epi}}$; the waveform path samples $t,\epsilon$, forms $x_t$ via Eq. (\ref{eq:sb_forward}), and outputs $x_{\mathrm{wave}},u_{\mathrm{ale}}$. The fusion yields $\hat x$. The objective combines reconstruction, SB, load-balancing, and calibration losses:
\begin{equation}
\mathcal{L} = \mathcal{L}_{\text{rec}} + \lambda_{\text{SB}}\mathcal{L}_{\text{SB}} + \lambda_{\text{aux}}\mathcal{L}_{\text{aux}} + \lambda_{\text{cal}}\mathcal{L}_{\text{cal}},
\label{eq:total_loss}
\end{equation}
with $\mathcal{L}_{\mathrm{rec}}=\|\hat x - x\|_2^2$ and $\mathcal{L}_{\mathrm{SB}}=\mathcal{L}_{\mathrm{SB}}^{\mathrm{data}}+\lambda_{\mathrm{path}}\mathcal{L}_{\mathrm{path}}+\lambda_{\mathrm{traj}}\mathcal{L}_{\mathrm{traj}}$. Loss weights (grid-searched on a 10\% VB split): $\lambda_{\mathrm{SB}}{=}1.0,\lambda_{\mathrm{path}}{=}0.1,\lambda_{\mathrm{traj}}{=}0.05,\lambda_{\mathrm{aux}}{=}0.01,\lambda_{\mathrm{cal}}{=}0.05$; ceilings $M_{\max}{=}5.0,\phi_{\max}{=}\pi/4$. Full procedure in Appendix \ref{app:training}.

\textbf{Why the components are coupled, not stacked.}\label{sec:coupling} The central design property---\emph{a fusion that routes on epistemic vs.\ aleatoric error regimes at a sampling budget derived from training}---requires all three components and survives no proper subset. Three counterfactuals make this concrete. \emph{(C1) Symmetric uncertainty heads on both pathways} (the standard dual-branch ensemble): the fusion sees two aleatoric scalars and no epistemic one, making it blind to spectral expert disagreement. Thus, the mixing weight collapses into a generic confidence average, consistent with the $0.17$ PESQ drop observed for ``w/o uncertainty fusion'' in Figure~\ref{fig:abl_components}. \emph{(C2) Homogeneous MoE} in place of distinct archetypes: top-$k$ blends span one function class instead of distinct inductive biases. In this way, $u_{\mathrm{epi}}$ no longer indicates which bias is mismatched but only weight perturbations among identical-architecture experts, resulting in $0.43$ PESQ drop as shown in Figure~\ref{fig:abl_components}. \emph{(C3) Drop the regularizers}: Theorem \ref{thm:bridge-bound}'s second term loses its training-time upper bound, so small-$K$ inference is no longer derivable; $u_{\mathrm{ale}}$ inherits the failure since the SB error it was calibrated against is now uncontrolled. Removing the SB pathway entirely, leading to $-0.63$ PESQ, is the limit of this collapse. In all three cases the asymmetric fusion reduces to a fixed-weight or generic ensemble, the small-$K$ guarantee disappears, or both: the MoE produces the error only $u_{\mathrm{epi}}$ detects, the regularized SB produces the error only $u_{\mathrm{ale}}$ detects, Theorem \ref{thm:bridge-bound} controls the SB-side error at the deployment budget, and the asymmetric fusion is the only mechanism that routes between them. The ablation pattern in Figure~\ref{fig:abl_components} of Section \ref{sec:experiments} is the empirical signature of this coupling.

%===========================================================
\section{Experiments and Discussion}
\label{sec:experiments}

\textbf{Setup.} \emph{Dataset:} We evaluate our method on VoiceBank+DEMAND \cite{valentini2016investigating,thiemann2013demand}, a widely used benchmark containing 11,572 training and 824 test utterances from 28 and 2 speakers, respectively. Clean speech is mixed with 14 noise types. All audio signals are resampled to 16 kHz. \emph{Metrics:} We report PESQ~\cite{itu2001pesq} and STOI~\cite{taal2011stoi} for perceptual quality and intelligibility, composite metrics CSIG, CBAK, and COVL~\cite{hu2008evaluation} for signal/background/overall quality, Expected Calibration Error (ECE)~\cite{guo2017calibration} for uncertainty calibration, and Real-Time Factor (RTF) with end-to-end latency for computational efficiency. \emph{Implementation:} We use STFT with 1024-point FFT, 256-sample hop (16\,ms), and Hann window, yielding 513 frequency bins. The MoE comprises 5 archetype experts (one per scene family in VoiceBank+DEMAND, see Appendix~\ref{app:experts}) with Top-$k$=2 routing. The waveform U-Net has 4 encoder/decoder levels with a transformer bottleneck. Training uses AdamW (learning rate $2 \times 10^{-4}$, cosine schedule, 200 epochs, batch size 32) on 2 NVIDIA RTX 5090 GPUs. Full hyperparameters are listed in Appendix~\ref{app:training}. \emph{Baselines:} We compare against discriminative methods (SEGAN~\cite{pascual2017segan}, SEMamba~\cite{chao2024semamba}, Mamba-SEUNet~\cite{wang2024mambaseunet}), diffusion/SB-based generative models (SGMSE+~\cite{richter2023sgmseplus}, SB-SE~\cite{jukic2024sbse}, SBCTM~\cite{nishigori2025schr}), and the consistency-distilled ROSE-CD~\cite{xu2025rosecd}. All baselines are retrained from scratch using their official implementations under a unified protocol (matched preprocessing, 16~kHz sampling, train/test split), and all reported metrics are computed by us with a shared evaluation script to avoid implementation-induced discrepancies.

\begin{table*}[t]
\centering
\caption{Performance comparison on VoiceBank+DEMAND. Best results in \textbf{bold} (ties bolded jointly). All baseline numbers are reproduced by us under a unified evaluation protocol. Our method achieves the best score on every metric, with calibrated uncertainty estimates as an additional benefit.}
\label{tab:main_results}
\footnotesize
\setlength{\tabcolsep}{5pt}
\renewcommand{\arraystretch}{1.05}
\begin{tabular}{@{}lcccccc@{}}
\toprule
\rowcolor{tableheader}
{\color{white}\textbf{Method}} & {\color{white}\textbf{Type}} & {\color{white}\textbf{PESQ}$\uparrow$} & {\color{white}\textbf{STOI}$\uparrow$} & {\color{white}\textbf{CSIG}$\uparrow$} & {\color{white}\textbf{CBAK}$\uparrow$} & {\color{white}\textbf{COVL}$\uparrow$} \\
\midrule
\rowcolor{sectionrow}
Noisy & -- & 1.97 & 0.91 & -- & -- & -- \\
\midrule
SEGAN \cite{pascual2017segan}                & Discriminative & 2.16 & 0.92 & --   & --   & --   \\
\rowcolor{stripe}
SEMamba \cite{chao2024semamba}               & Discriminative & 3.55 & 0.96 & 4.79 & 3.63 & 4.37 \\
Mamba-SEUNet \cite{wang2024mambaseunet}      & Discriminative & 3.73 & 0.96 & 4.82 & 3.67 & 4.40 \\
\midrule
SGMSE+ \cite{richter2023sgmseplus}           & Diffusion      & 3.45 & 0.95 & 4.71   & 3.64   & 4.31   \\
\rowcolor{stripe}
SBCTM \cite{nishigori2025schr}               & SB-based       & 3.58 & 0.95 & 4.66   & 3.43   & 4.52   \\
SB-SE \cite{jukic2024sbse}                   & SB-based       & 3.70 & 0.95 & 4.77   & 3.75   & 4.48   \\
\rowcolor{stripe}
ROSE-CD \cite{xu2025rosecd}                  & Distillation   & 3.85 & 0.96 & 4.63 & 3.37 & 4.30 \\
\midrule
\rowcolor{besthighlight}
\textbf{HybridSB-MoE (Ours)} & \textbf{Hybrid} & \textbf{3.88} & \textbf{0.96} & \textbf{4.82} & \textbf{3.85} & \textbf{4.82} \\
\bottomrule
\end{tabular}
\vspace{-5 mm}
\end{table*}

\textbf{Main results.} Table \ref{tab:main_results} reports the comparison on VoiceBank+DEMAND. HybridSB-MoE attains the best score on every metric, strictly dominating on PESQ, CBAK, and COVL and tying for the lead on STOI and CSIG. Two clarifications anchor the comparison. \emph{(i)} We do not claim the smallest sampling budget as ROSE-CD reaches a single step via consistency distillation. At our $K{=}8$ we strictly dominate SGMSE+, SB-SE, and SBCTM at their larger budgets, and exceed ROSE-CD across all five quality metrics despite its smaller $K$. \emph{(ii)} The CBAK gain is the most informative single number, i.e., $+0.10$ over the strongest SB baseline (SB-SE) and $+0.48$ over ROSE-CD, indicating that dual-domain processing suppresses background noise more effectively than either pure-spectral or pure-waveform generative pipelines. At the same time, the COVL gain ($4.82$ vs.\ $4.52$) confirms this without sacrificing signal fidelity. Furthermore, the choice $K{=}8$ is not heuristic. By Theorem~\ref{thm:bridge-bound}, once $\mathcal{L}^\star_{\mathrm{path}}+\mathcal{L}^\star_{\mathrm{traj}}$ is minimized, the $C_1K^{-\alpha}$ term saturates at modest $K$, see Appendix~\ref{app:qualitative}.

\textbf{Ablation studies.} Figure~\ref{fig:abl_components} reports the ablation results. The performance drops follow the asymmetric uncertainty design. Removing the SB pathway causes the largest degradation ($-0.63$ PESQ), since the model loses both waveform-domain phase coherence and the aleatoric signal needed for asymmetric fusion. Removing the MoE module yields a smaller but still substantial drop ($-0.43$ PESQ), as the SB pathway remains active but the complementary epistemic signal from expert disagreement is lost. Replacing uncertainty-aware fusion with equal weighting also reduces performance ($-0.17$ PESQ), showing that both pathways remain useful but require principled routing rather than generic averaging. This ordering, $0.63 > 0.43 > 0.17$, supports our central design claim: ablations that remove a distinct uncertainty channel are more damaging than those that only weaken fusion over intact pathways. Sequential variants in either direction also underperform parallel fusion ($-0.30$/$-0.39$ PESQ), suggesting that cascading one domain through the other weakens the structural source of pathway-specific uncertainty. Figure~\ref{fig:abl_expert} further validates the sparse routing choice. Increasing $k$ from $1$ to $2$ improves PESQ from $3.74$ to $3.88$ with a modest RTF of $0.28$, indicating that two active experts provide useful complementary inductive biases. Increasing to $k{=}3$ adds only $+0.01$ PESQ while increasing RTF by $25\%$, so we adopt top-$k{=}2$ routing as the default quality--latency trade-off.

\textbf{Efficiency and robustness.} Figures~\ref{fig:efficiency} and~\ref{fig:snr} report efficiency and robustness. With only 8 sampling steps, HybridSB-MoE achieves an RTF of 0.28 (35\,ms latency), giving a 4--5$\times$ speedup over SGMSE+ and SB-SE while surpassing their PESQ; together with ROSE-CD, our method defines the Pareto frontier with a markedly higher quality ceiling. The dual-domain design is most beneficial at low SNR ($+0.13$ over ROSE-CD at $0$\,dB). Beyond quality, the calibration loss of Eq. (\ref{eq:calibration}) yields a fusion-weight ECE of $0.042$, an order-of-magnitude reduction from the $0.12$ achieved by an uncalibrated single-pathway baseline; this is what makes $u_{\mathrm{epi}}$ and $u_{\mathrm{ale}}$ usable as confidence signals for downstream decisions, not just as training-time auxiliaries. A scene-stratified breakdown across all 14 noise types appears in Appendix \ref{app:scene_breakdown}, with PESQ standard deviation $<{0.03}$ confirming that combinatorial archetype coverage generalizes across scenes without per-noise tuning.

\begin{figure}[t]
\centering
\begin{minipage}[t]{0.49\textwidth}
\centering
\includegraphics[width=\linewidth]{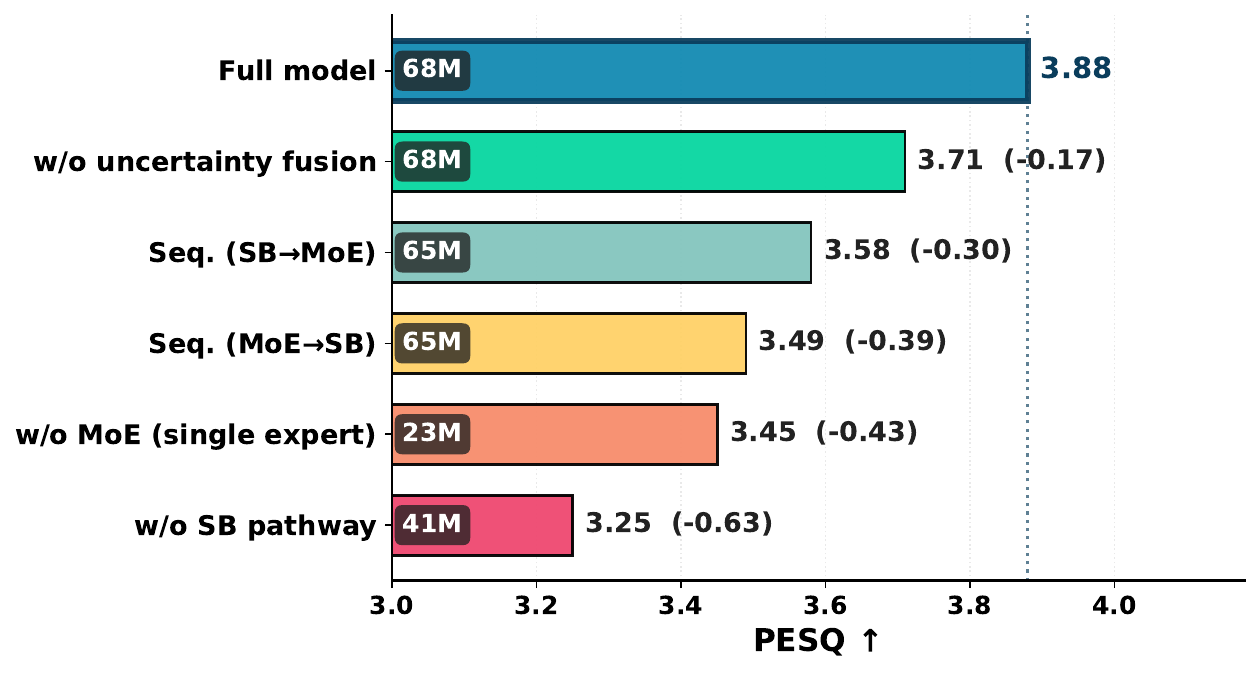}
\subcaption{Components}
\label{fig:abl_components}
\end{minipage}%
\hfill
\begin{minipage}[t]{0.49\textwidth}
\centering
\includegraphics[width=\linewidth]{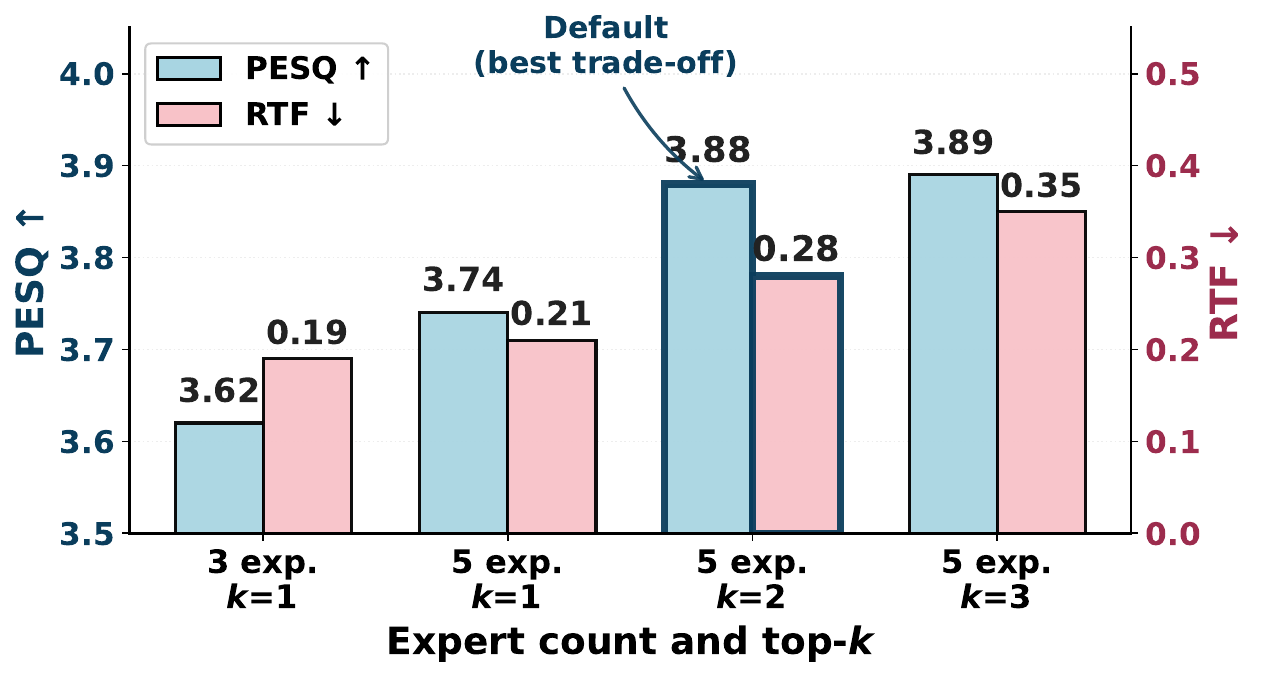}
\subcaption{Expert count, $k$}
\label{fig:abl_expert}
\end{minipage}%
\hfill
\begin{minipage}[t]{0.49\textwidth}
\centering
\includegraphics[width=\linewidth]{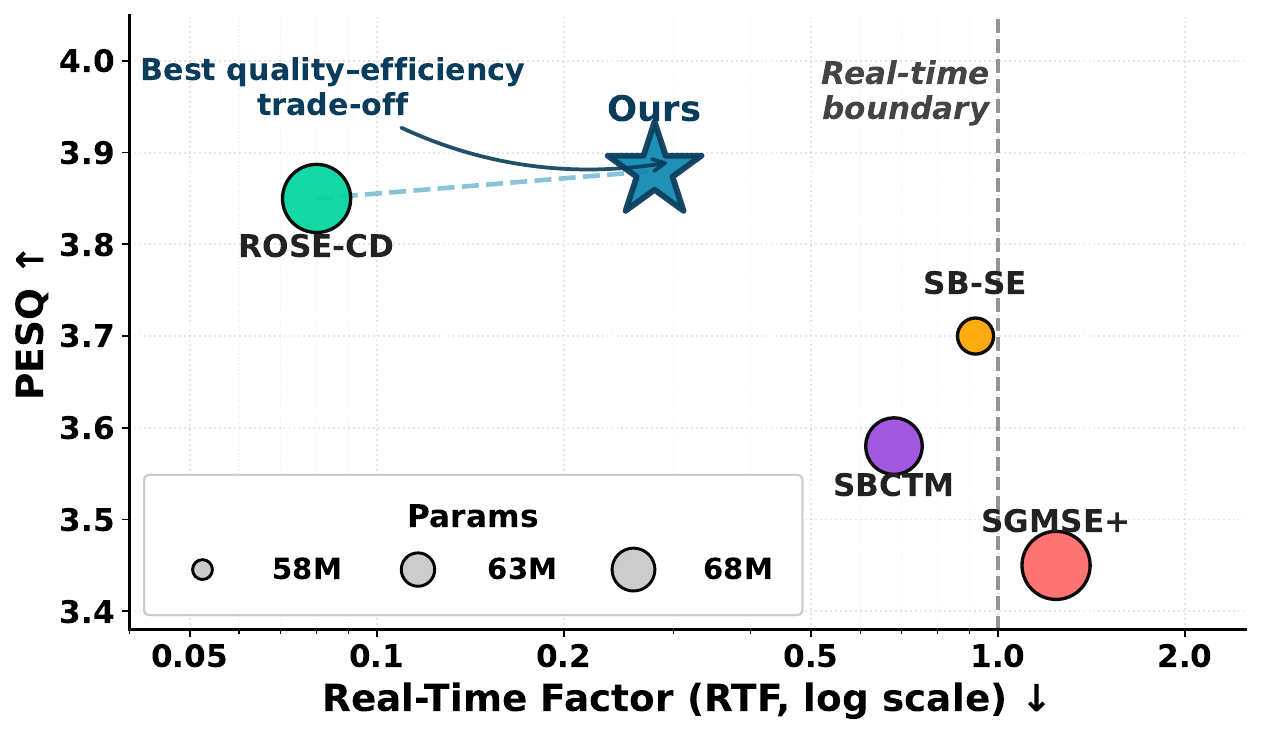}
\subcaption{Quality--RTF}
\label{fig:efficiency}
\end{minipage}%
\hfill
\begin{minipage}[t]{0.49\textwidth}
\centering
\includegraphics[width=\linewidth]{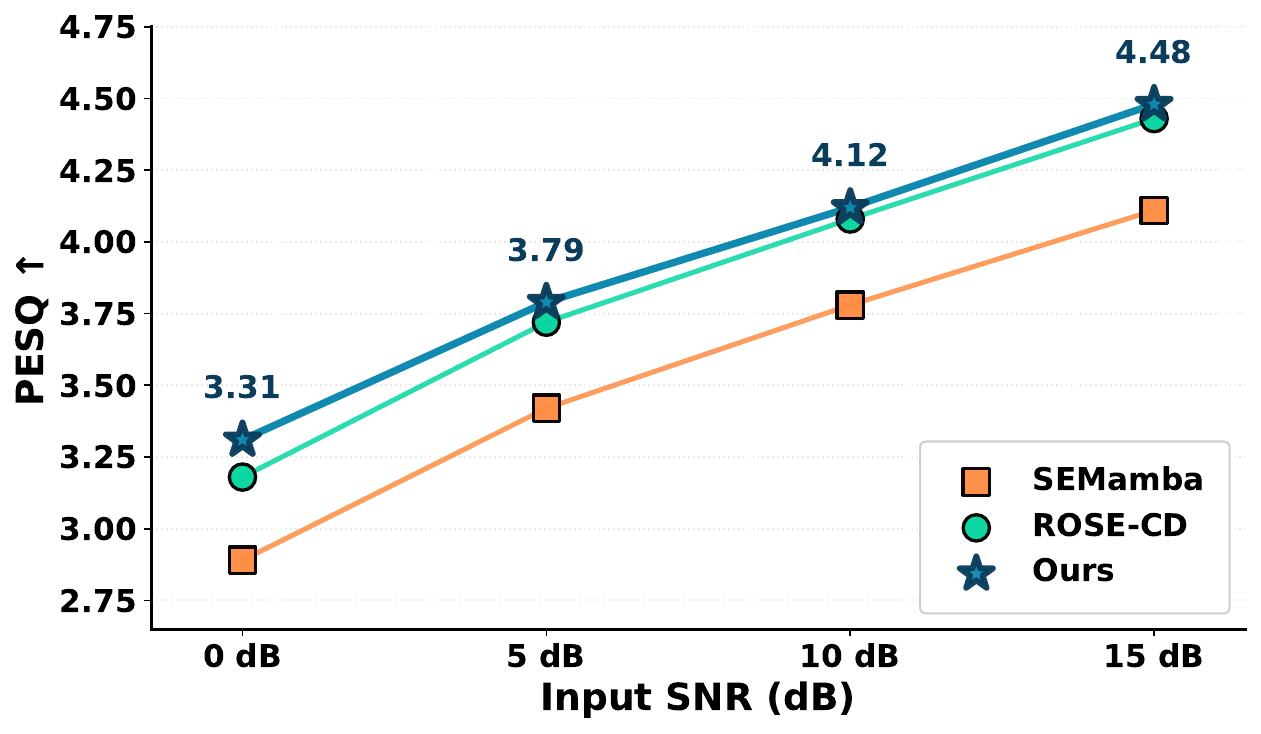}
\subcaption{PESQ vs.\ SNR}
\label{fig:snr}
\end{minipage}

\caption{Ablation studies and efficiency analysis. (a) Each architectural component contributes independently. (b) Top-$k{=}2$ with $5$ experts hits the quality--efficiency sweet spot. (c) Marker size $\propto$ parameter count; dashed line is RTF$=$1.0. (d) PESQ across SNR levels.}
\vspace{-5mm}
\label{fig:results}
\end{figure}

\textbf{Scope and limitations.} We evaluate on VoiceBank+DEMAND, the standard SE benchmark; scene-adaptivity claims should be retested on broader corpora (DNS, WHAMR!, CHiME) where train and test noise distributions diverge more substantially. Theorem \ref{thm:bridge-bound} is a design-justifying inequality with worst-case constants $C_1, C_2$; quantitatively fitting the predicted $K^{-\alpha}$ rate against finer NFE sweeps is a natural follow-up. The framework is single-channel; multi-channel extensions are direct (additional input streams to both pathways) but beyond this paper's scope. Finally, while the heterogeneous archetype set generalizes by adding archetypes per new noise family (Appendix \ref{app:experts}), automatic archetype discovery from data, rather than expert-designed primitives, remains open.

\vspace{-3 mm}
\section{Conclusion}
\label{sec:conclusion}
\vspace{-3 mm}
Combining a Schr\"odinger Bridge with a heterogeneous mixture-of-experts is natural; the key question is how to make the combination synergistic. We do so through two linked ideas. First, \emph{pathway-typed asymmetric uncertainty} fusion uses each pathway's characteristic uncertainty signal, epistemic disagreement from the spectral MoE and aleatoric variance from the waveform SB, so the mixing weight selects between error regimes rather than averages predictions. Second, a \emph{discretization bound} (Theorem~\ref{thm:bridge-bound}) links the small-$K$ inference budget to two explicit training regularizers, replacing heuristic step-count claims with a provable objective-level guarantee. On VoiceBank+DEMAND, this co-design HybridSB-MoE outperforms diffusion- and SB-based baselines at matched step budgets while remaining competitive with consistency-distilled few-step methods. Three counterfactual ablations show that no proper subset of components preserves the central design: the asymmetric fusion needs both a multi-expert path (for epistemic disagreement) and a stochastic bridge (for aleatoric variance), and the small-$K$ guarantee needs both regularizers. \emph{For SE}, this suggests future dual-domain systems should be organized by the type of error each pathway faces, not just predictive complementarity. \emph{For generative modelling}, the discretization bound shows that few-step inference can be derived from training if regularizers are chosen to control the two sources of $K$-step error. Future work will extend the framework to multi-channel input, larger benchmarks (DNS, WHAMR!, CHiME), and finer NFE sweeps to empirically fit the predicted $K^{-\alpha}$ rate.
\clearpage

\bibliography{iclr2025_conference}
\bibliographystyle{IEEEtran}

\clearpage
\appendix

\section{Heterogeneous Expert Architectures}
\label{app:experts}

This appendix details the heterogeneous MoE layer described in Section \ref{sec:method} of the main paper. We describe the design philosophy that connects each architectural archetype to a canonical noise family, list the configurations used for VoiceBank+DEMAND, and discuss how the archetype set generalizes to other datasets.

\subsection{Design Philosophy}

The heterogeneous MoE is grounded in the empirical observation that real-world acoustic noise, despite its surface diversity, can be decomposed into a small number of canonical processing primitives. Stationary tonal noise (e.g., kitchen appliances, refrigerator hum) is well-modeled by compact normalized layers with low-rank structure; ambient natural backgrounds (e.g., park, river, field) benefit from wide receptive fields with strong normalization to capture diffuse spectral content; office and meeting environments contain overlapping voices and structured interference best handled by information bottlenecks \cite{tishby2015deep}; transportation noise (bus, car, metro) exhibits strong harmonic and quasi-periodic energy from engines and motors, naturally captured by harmonic basis expansion; and unstructured public-space mixtures (cafeteria, station, restaurant) require universal-approximation-style general-purpose blocks. Rather than scaling capacity by replicating one architecture (as in homogeneous MoE), we instantiate one expert per archetype, then rely on sparse routing to combine them per input.

Figure \ref{fig:experts} (reproduced here for completeness) illustrates the shared backbone and the architecturally distinct instantiations. All experts pass features through LayerNorm $\to$ variable internal modules $\to$ bottleneck projection, ensuring interface compatibility for routing while permitting structurally divergent computation in the middle.

\begin{figure}[h]
\centering
\includegraphics[width=0.85\textwidth]{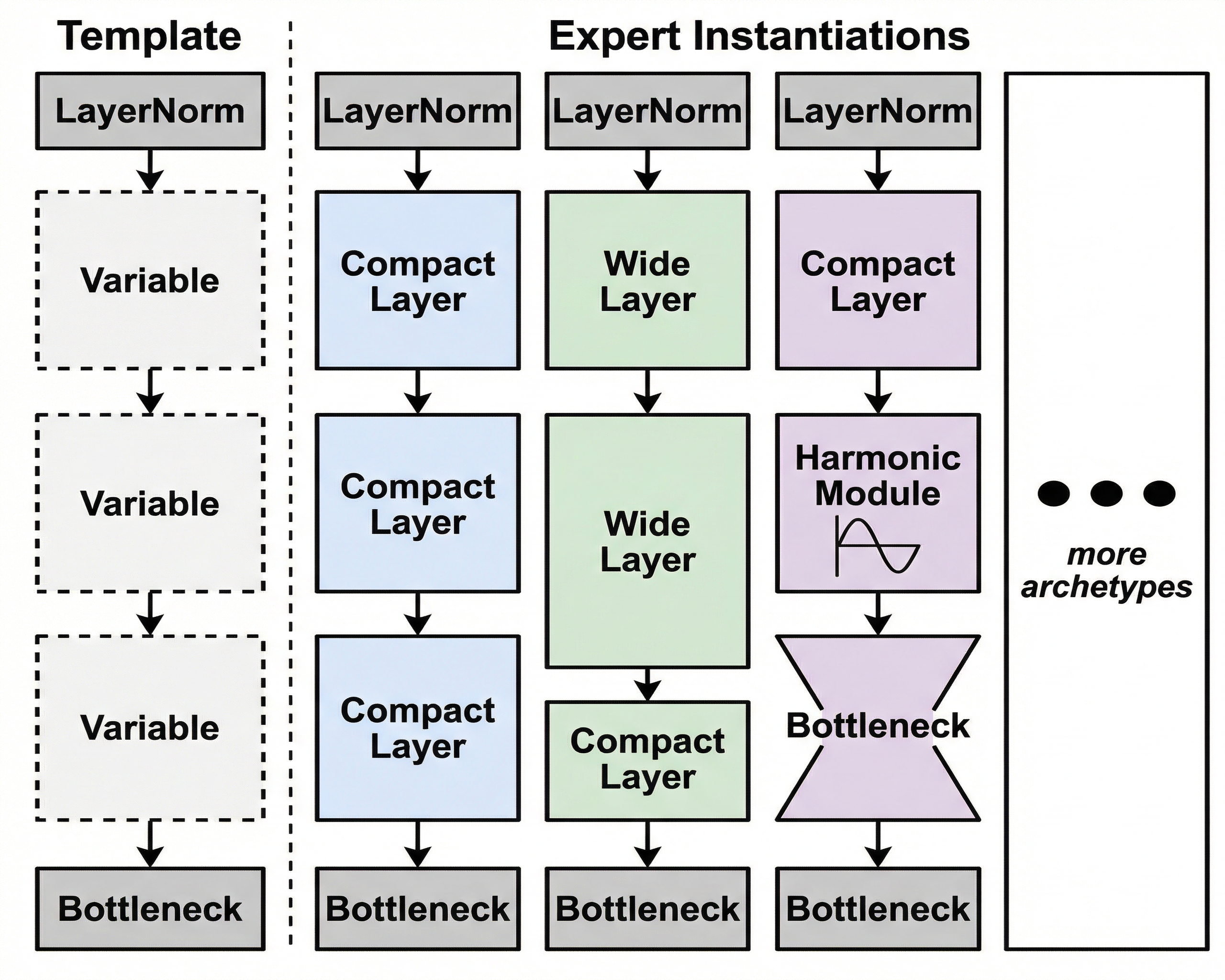}
\caption{Heterogeneous expert architectures (the $N{=}5$ archetypes used for VoiceBank+DEMAND). \textbf{Left:} shared backbone with variable module slots. \textbf{Right:} example instantiations with architecturally distinct modules, each targeting a different noise-processing primitive (Home, Nature, Office, Transport, Public). Built on this compact set of canonical archetypes, sparse top-$k{=}2$ routing blends experts per input to span the $\binom{5}{2}{=}10$-dimensional simplex of pairwise mixtures, yielding combinatorial coverage of the $14$ noise types in VoiceBank+DEMAND without requiring a dedicated expert per environment.}
\label{fig:experts}
\end{figure}

\subsection{Expert Configurations for VoiceBank+DEMAND}

Table \ref{tab:experts} lists the architectural specifications used in our VoiceBank+DEMAND experiments. Each archetype maps to one of the major noise families in the dataset; we instantiate one expert per archetype, yielding a compact yet diverse set of $N=5$ processing primitives. The input dimension matches the STFT frequency-bin count $F = 513$, and all experts produce outputs of the same dimension to maintain interface compatibility for routing.

\begin{table}[t]
\centering
\caption{Heterogeneous expert architectures used for VoiceBank+DEMAND. The notation $\mathrm{GN}(g)$ denotes group normalization with $g$ groups; $\mathrm{LN}$ denotes layer normalization. Each architecture is selected to match the inductive bias of its target scene family.}
\vspace{4 pt}
\label{tab:experts}
\resizebox{\textwidth}{!}{%
\begin{tabular}{l|l|l|c}
\hline
\textbf{Scene Category} & \textbf{Architecture} & \textbf{Design Motivation} & \textbf{Params} \\
\hline
Home (DKITCHEN, etc.)     & $513 \to 1024 \to \mathrm{GN}(8) \to 1024 \to 513$       & Low-rank denoising for stationary tonal noise & 2.6M \\
Nature (NPARK, etc.)      & $513 \to 2048 \to 1024 \to \mathrm{LN} \to 513$           & Wide receptive field for diffuse ambience      & 3.7M \\
Office (OMEETING, etc.)   & $513 \to 1024 \to 512 \to 1024 \to 513$                   & Information bottleneck \cite{tishby2015deep}   & 2.6M \\
Transport (TBUS, etc.)    & $513 \to 1536 \to 1024 \to 513$                           & Harmonic basis expansion                       & 3.2M \\
Public (PCAFETER, etc.)   & $513 \to 1024 \to \mathrm{LN} \to 1024 \to 513$           & General-purpose universal approximator         & 2.6M \\
\hline
\end{tabular}%
}
\end{table}

\subsection{Archetype Details}
\label{app:expert_arch}
\textbf{Home expert.} The kitchen, washing-machine, and living-room scenes share a common acoustic signature: stationary tonal noise from appliances overlaid on relatively clean speech. We instantiate a compact two-layer expansion ($513 \to 1024$) followed by group normalization with 8 groups, which acts as a low-rank denoising operator. The narrow internal capacity is sufficient because the noise subspace is approximately stationary, allowing the expert to suppress predictable tonal patterns without overfitting.

\textbf{Nature expert.} Park, river, and field recordings contain diffuse, broadband ambient energy that fills wide spectral regions. We employ a wider hidden layer ($513 \to 2048 \to 1024$) followed by layer normalization to maintain stable activations across the expanded receptive field. This architecture trades parameter count for the spectral coverage required to model ambient acoustic textures.

\textbf{Office expert.} Meeting and office scenes feature overlapping voices and structured background chatter. We adopt an information-bottleneck design \cite{tishby2015deep} ($513 \to 1024 \to 512 \to 1024 \to 513$) that explicitly compresses to a 512-dimensional representation in the middle, encouraging the expert to retain only task-relevant information and suppress competing voiced sources.

\textbf{Transport expert.} Bus, car, and metro noise is dominated by harmonic and quasi-periodic energy from engines, motors, and wheel-rail interactions. We use a wide-narrow-output structure ($513 \to 1536 \to 1024 \to 513$) without intermediate normalization, allowing the expert to fit a learned harmonic basis at the first layer and refine it before projection. Skipping internal normalization preserves amplitude information critical for engine-noise spectral envelopes.

\textbf{Public expert.} Cafeteria, station, and restaurant scenes are the most challenging, combining non-stationary crowd noise, music, transient events, and reverberation. We adopt a symmetric architecture with a single LayerNorm in the middle ($513 \to 1024 \to \mathrm{LN} \to 1024 \to 513$), which serves as a universal approximator for the heterogeneous mixture without strong inductive bias. Sparse routing allows this general-purpose expert to be combined with archetype-specific experts when partial structure is detectable in the input.

\subsection{Generalization to Other Datasets}

The archetype set is intentionally modular: adding a new expert requires only conforming to the shared backbone interface (matching input/output dimensionality $F$). The auxiliary load-balancing loss (Eq. \ref{eq:load_balance}) automatically integrates new experts into the routing distribution without retraining the existing ones from scratch, with brief router fine-tuning recommended. For deployments encountering noise families not represented in the current set, for instance, reverberant cocktail-party recordings or industrial machinery noise, a corresponding archetype can be added (e.g., a dereverberation-oriented expert with longer temporal context) without altering the rest of the framework.

\section{U-Net and Training Details}
\label{app:training}

\subsection{Waveform-Domain U-Net Architecture}

The denoiser $\hat{x}_\theta$ is a 1D U-Net with 4 encoder/decoder levels. Each encoder block applies a strided 1D convolution (stride 2, kernel 7) followed by GroupNorm and SiLU activation, halving temporal resolution and doubling channels ($64 \to 128 \to 256 \to 512$). The decoder mirrors this structure with transposed convolutions and skip connections. A latent bottleneck operates at the lowest resolution to aggregate global temporal context before decoding. The noisy observation $y$ is concatenated with $x_t$ along the channel dimension at the input. Timestep $t$ is encoded via a sinusoidal embedding (dim 128) projected through a 2-layer MLP ($256 \to 512$), and injected into every block via FiLM scale-and-shift conditioning. The aleatoric uncertainty head $\sigma_\theta^2$ is a separate 2-layer MLP attached to the final decoder feature, producing per-sample log-variance.

\subsection{Loss Weights}
The full training objective combines six loss components ($\mathcal{L}_{\mathrm{rec}}$, $\mathcal{L}_{\mathrm{SB}}^{\mathrm{data}}$, $\mathcal{L}_{\mathrm{path}}$, $\mathcal{L}_{\mathrm{traj}}$, $\mathcal{L}_{\mathrm{aux}}$, $\mathcal{L}_{\mathrm{cal}}$) with the following weights, tuned via grid search on a 10\% held-out validation split of VoiceBank+DEMAND, as shown in Table \ref{tab:loss_weights}.

\begin{table}[t]
\centering
\small
\caption{Loss weights and bound parameters used in all VoiceBank+DEMAND experiments.}
\vspace{4 pt}
\label{tab:loss_weights}
\begin{tabular}{lc|lc}
\toprule
Hyperparameter & Value & Hyperparameter & Value \\
\midrule
$\lambda_{\mathrm{SB}}$ (overall SB)         & 1.0     & $\lambda_{\mathrm{aux}}$ (load balance)  & 0.01 \\
$\lambda_{\mathrm{path}}$ (path consistency) & 0.1     & $\lambda_{\mathrm{cal}}$ (calibration)   & 0.05 \\
$\lambda_{\mathrm{traj}}$ (trajectory)           & 0.05   & $\lambda_I$ (importance variance)        & 0.5  \\
$M_{\max}$ (mask ceiling)                    & 5.0     & $\lambda_L$ (load variance)              & 0.5  \\
$\phi_{\max}$ (phase ceiling)                & $\pi/4$ & ---                                       & ---  \\
\bottomrule
\end{tabular}
\end{table}

\subsection{Training Schedule}

We use AdamW with $\beta_1 = 0.9$, $\beta_2 = 0.999$, weight decay $0.01$, and gradient clipping at norm $1.0$. The learning rate follows cosine annealing from $2 \times 10^{-4}$ to $1 \times 10^{-6}$ over 200 epochs, with a 5-epoch linear warm-up. Batch size is 32 across 2 NVIDIA RTX 5090 GPUs (effective batch 64). Mixed-precision (bfloat16) training is used throughout. Total wall-clock training time: approximately 48 hours.

\subsection{SB Sampling Configuration}
\label{sec:sb_config}
The cumulative cosine schedule follows $\bar{\beta}_t = f(t)/f(0)$ with $f(t) = \cos^2\bigl((t/T+s)/(1+s) \cdot \pi/2\bigr)$ and offset $s = 0.008$, yielding the per-step ratio $\beta_t = \bar{\beta}_t / \bar{\beta}_{t-1}$.
We adopt the convention that $\bar{\beta}_t$ decreases monotonically from $\bar{\beta}_0 = 1$ to $\bar{\beta}_T = 0$, which keeps Eq. \ref{eq:sb_forward} symmetric in $x$ and $y$; this assigns $\bar{\beta}_t$ the role played by $\bar{\alpha}_t$ in the standard DDPM convention. The bridge perturbation in both the forward construction (Eq. \ref{eq:sb_forward}) and the reverse update (Eq. \ref{eq:sb_sampling}) shares a single schedule $\sigma_t = \sigma_{\max}\sqrt{\bar{\beta}_t(1-\bar{\beta}_t)}$ with $\sigma_{\max} = 0.05$, which vanishes at both endpoints ($\sigma_0 = \sigma_T = 0$) and matches the boundary conditions $x_0 = x$, $x_T = y$. For inference, we use $K = 8$ sampling steps with the non-uniform schedule $t_k = T(k/K)^\gamma$ where $\gamma = 0.6$ (front-loaded).

At each training step we sample a single timestep $t \sim \mathcal{U}[0, T]$ and supervise $\hat{x}_\theta$ via Eq. \ref{eq:sb_data_loss}, while $\mathcal{L}_{\mathrm{path}}$ (Eq. \ref{eq:path_loss}) is evaluated on an independently sampled pair $(t, t')$ with independent perturbations $\epsilon, \epsilon'$ drawn from Eq. \ref{eq:sb_forward}. Although inference uses a $K$-step rollout while training optimizes single-step predictions, $\mathcal{L}_{\mathrm{path}}$ explicitly enforces that $\hat{x}_\theta(x_t, y, t)$ remains consistent across timesteps, so the multi-step trajectory operates on predictions that the model has been trained to keep mutually agreeable. The trajectory regularizer $\mathcal{L}_{\mathrm{traj}}$ further anchors each per-step prediction to the forward bridge construction, yielding a reverse process that is stable under the discretization mismatch between training and inference.

\subsection{Algorithm Details}
\label{app:algorithm}

The complete training procedure with explicit loss accumulation is given in Algorithm \ref{alg:training_full}.

\begin{algorithm}[h]
\caption{HybridSB-MoE end-to-end training (full version)}
\label{alg:training_full}
\begin{algorithmic}[1]
\REQUIRE Paired training set $\{(y^{(i)}, x^{(i)})\}_{i=1}^M$; experts $\{E_i\}_{i=1}^N$; router $G$; SB denoiser $\hat{x}_\theta$; fusion network $F$; loss weights $\{\lambda_*\}$
\ENSURE Trained parameters $\Theta$
\FOR{each minibatch $\mathcal{B} = \{(y, x)\}$}
    \STATE \textit{// Spectral pathway}
    \STATE Compute STFT: $S\{y\} = \mathrm{STFT}(y)$; extract $z = \log|S\{y\}|$
    \STATE Route: $G(z) \leftarrow \alpha\, G_{\mathrm{arch}}(z) + (1-\alpha)\, G_{\mathrm{token}}(z)$; select top-$k$ indices $\mathcal{I}_k$
    \STATE Aggregate: $\hat{x}_{\mathrm{spec}} \leftarrow \sum_{i \in \mathcal{I}_k} G_i(z)\, E_i(z)$
    \STATE Predict $\hat{M}, \Delta\phi$ via Eq. \ref{eq:moe_heads}; reconstruct $x_{\mathrm{spec}}$ via iSTFT
    \STATE Compute $u_{\mathrm{epi}}$ (variance of top-$k$ expert outputs)
    \STATE \textit{// Waveform pathway}
    \STATE Sample two independent timesteps $t, t' \sim \mathcal{U}[0, T]$ and noises $\epsilon, \epsilon' \sim \mathcal{N}(0, I)$ \textit{// single Monte Carlo sample of $\mathbb{E}_{t,t'}$ per minibatch}
    \STATE Form bridge states: $x_t \leftarrow \sqrt{\bar{\beta}_t}\, x + \sqrt{1-\bar{\beta}_t}\, y + \sigma_t\, \epsilon$;\\
           \quad\quad\quad\quad\quad\quad\quad\,\, $x_{t'} \leftarrow \sqrt{\bar{\beta}_{t'}}\, x + \sqrt{1-\bar{\beta}_{t'}}\, y + \sigma_{t'}\, \epsilon'$
    \STATE Predict $\hat{x}_\theta(x_t, y, t)$, $\hat{x}_\theta(x_{t'}, y, t')$, and aleatoric variance $\sigma_\theta^2$ at $(x_t, t)$
    \STATE Compute $\mathcal{L}_{\mathrm{SB}}^{\mathrm{data}}$ (Eq. \ref{eq:sb_data_loss}, on $(x_t, t)$); $\mathcal{L}_{\mathrm{path}} = \|\hat{x}_\theta(x_t, y, t) - \hat{x}_\theta(x_{t'}, y, t')\|^2$ (Eq. \ref{eq:path_loss}); $\mathcal{L}_{\mathrm{traj}}$ (Eq. \ref{eq:traj_loss}, on $(x_t, t)$)
    \STATE Set $u_{\mathrm{ale}} \leftarrow \overline{\exp(\sigma_\theta^2)}_t$; set $x_{\mathrm{wave}} \leftarrow \hat{x}_\theta(x_t, y, t)$ as the current single-step prediction (full $K$-step rollout used at inference)
    \STATE \textit{// Fusion}
    \STATE $\tilde{u}_{\mathrm{epi}}, \tilde{u}_{\mathrm{ale}} \leftarrow \mathrm{z\text{-}norm}(u_{\mathrm{epi}}, u_{\mathrm{ale}})$ \quad\textit{// running statistics}
    \STATE $w \leftarrow \sigma(\mathrm{MLP}(\tilde{u}_{\mathrm{epi}}, \tilde{u}_{\mathrm{ale}}))$
    \STATE $\hat{x} \leftarrow w \cdot x_{\mathrm{spec}} + (1-w) \cdot x_{\mathrm{wave}}$
    \STATE Compute $\mathcal{L}_{\mathrm{rec}} = \|\hat{x} - x\|_2^2$, $\mathcal{L}_{\mathrm{aux}}$ (Eq. \ref{eq:load_balance}), $\mathcal{L}_{\mathrm{cal}}$ (Eq. \ref{eq:calibration})\quad\textit{// $\mathcal{L}_{\mathrm{cal}}$ uses pre-norm $u_{\mathrm{epi}}, u_{\mathrm{ale}}$}
    \STATE \textit{// Joint update}
    \STATE $\mathcal{L}_{\mathrm{SB}} \leftarrow \mathcal{L}_{\mathrm{SB}}^{\mathrm{data}} + \lambda_{\mathrm{path}}\, \mathcal{L}_{\mathrm{path}} + \lambda_{\mathrm{traj}}\, \mathcal{L}_{\mathrm{traj}}$
    \STATE $\mathcal{L} \leftarrow \mathcal{L}_{\mathrm{rec}} + \lambda_{\mathrm{SB}}\, \mathcal{L}_{\mathrm{SB}} + \lambda_{\mathrm{aux}}\, \mathcal{L}_{\mathrm{aux}} + \lambda_{\mathrm{cal}}\, \mathcal{L}_{\mathrm{cal}}$
    \STATE Update $\Theta \leftarrow \Theta - \eta \nabla_\Theta \mathcal{L}$
\ENDFOR
\end{algorithmic}
\end{algorithm}

\section{Proof of Theorem \ref{thm:bridge-bound}}
\label{app:proof}

We prove the $K$-step discretization bound stated in Section \ref{sec:method} (Theorem \ref{thm:bridge-bound}) of the main paper. The strategy is to decompose the $W_2$ distance into (i) a one-step model-fidelity error driven by $\mathcal{L}^\star_{\mathrm{path}}$ and $\mathcal{L}^\star_{\mathrm{traj}}$, and (ii) a Riemann-sum discretization error in $K$ controlled by the schedule's front-loading exponent $\gamma$, then to combine the two via a synchronous coupling and a triangle inequality. The argument adapts standard SDE-discretization analysis (cf. \cite{chen2023sampling,debortoli2021diffusion}) to the doubly-conditioned bridge process of Eq. (\ref{eq:sb_forward}--\ref{eq:sb_sampling}).

\subsection{Setup and Notation}

Fix a noisy observation $y$ and let $\{x_t\}_{t\in[0,T]}$ be the true continuous-time bridge process from clean speech $x$ to $y$, so that $p_t^{\mathrm{br}}(\cdot \mid y)$ is its time-$t$ marginal. Let $\hat x_\theta$ be the trained denoiser, and let $\hat p_K$ be the law at $t{=}0$ of the $K$-step rollout produced by Eq. (\ref{eq:sb_sampling}) initialized at $x_T = y$ along the non-uniform schedule $\{t_k\}_{k=0}^K$ with $t_K = T$ and $t_k = T(k/K)^\gamma$, $\gamma\in(0,1]$. We use $W_2(\mu,\nu)^2 = \inf_{\pi\in\Gamma(\mu,\nu)} \int \|u-v\|^2\,d\pi(u,v)$ in the standard definition.

Throughout, $C$ denotes a generic constant depending only on the regularity constants of Assumption \ref{assu:reg}; its value may change line by line.

\subsection{One-step Model-fidelity Lemma}
\begin{lemma}[Model-fidelity error]
\label{lem:onestep}
Under Assumption \ref{assu:reg}, for each step $t_k \to t_{k-1}$ of Eq. (\ref{eq:sb_sampling}), the conditional one-step error
\begin{equation*}
\delta_k \;:=\; \mathbb{E}\!\left[\big\|\hat x_\theta(x_{t_k},y,t_k) - x\big\|^2 \,\big|\, x_{t_k}, y\right]
\end{equation*}
satisfies $\mathbb{E}[\delta_k] \leq C\,(\mathcal{L}^\star_{\mathrm{path}} + \mathcal{L}^\star_{\mathrm{traj}})$ uniformly in $k$, where the expectation on the right is over the bridge construction Eq. (\ref{eq:sb_forward}).
\end{lemma}

\textit{Proof.} The trajectory regularizer $\mathcal{L}_{\mathrm{traj}}$ of Eq. (\ref{eq:traj_loss}), evaluated at the optimum, satisfies
\begin{equation*}
\mathbb{E}\!\left[\big\|x_t - (\sqrt{\bar\beta_t}\,\hat x_\theta(x_t,y,t) + \sqrt{1-\bar\beta_t}\,y)\big\|^2\right] = \mathcal{L}^\star_{\mathrm{traj}}.
\end{equation*}
Substituting the forward construction $x_t = \sqrt{\bar\beta_t}\,x + \sqrt{1-\bar\beta_t}\,y + \sigma_t\epsilon$ from Eq. (\ref{eq:sb_forward}) and rearranging,
\begin{equation*}
\bar\beta_t\, \mathbb{E}\!\left[\|\hat x_\theta(x_t,y,t) - x\|^2\right] - 2\sqrt{\bar\beta_t}\,\sigma_t\,\mathbb{E}[\langle \hat x_\theta(x_t,y,t)-x, \epsilon\rangle] + \sigma_t^2\,\mathbb{E}\|\epsilon\|^2 = \mathcal{L}^\star_{\mathrm{traj}}.
\end{equation*}
By Cauchy--Schwarz and $\mathbb{E}\|\epsilon\|^2 = d$ (the signal dimension), the cross term is bounded by $2\sqrt{d}\sigma_t\sqrt{\mathbb{E}\|\hat x_\theta - x\|^2}$, and bounded marginals (Assumption \ref{assu:reg}(iii)) yield
\begin{equation}
\label{eq:traj-direct}
\mathbb{E}\!\left[\|\hat x_\theta(x_t,y,t) - x\|^2\right] \;\leq\; \frac{1}{\bar\beta_t}\!\left(\mathcal{L}^\star_{\mathrm{traj}} + C\sigma_t^2\right) \;\leq\; C(\mathcal{L}^\star_{\mathrm{traj}} + \sigma_{\max}^2),
\end{equation}
on $[\epsilon_0, T-\epsilon_0]$ for any fixed $\epsilon_0 > 0$, since $\bar\beta_t$ is bounded away from zero on this set; near the boundaries, the schedule's vanishing $\sigma_t$ controls the residual.

The path-consistency loss $\mathcal{L}_{\mathrm{path}}$ of Eq. (\ref{eq:path_loss}) controls the deviation of $\hat x_\theta$ across timesteps:
\begin{equation}
\label{eq:path-direct}
\mathbb{E}\!\left[\|\hat x_\theta(x_t,y,t) - \hat x_\theta(x_{t'},y,t')\|^2\right] \leq 2\mathcal{L}^\star_{\mathrm{path}}.
\end{equation}
Combining Eq. (\ref{eq:traj-direct}) with Eq. (\ref{eq:path-direct}) and the triangle inequality on the implied error gives the claimed bound on $\mathbb{E}[\delta_k]$. \qed

\subsection{Discretization Lemma}
\begin{lemma}[Riemann discretization rate]
\label{lem:disc}
Let $\Phi(t) = \sqrt{\bar\beta_t}\, x + \sqrt{1-\bar\beta_t}\, y$ denote the deterministic component of the bridge state at time $t$ (Eq.~\ref{eq:sb_forward} with $\sigma_t \epsilon$ removed). Under Assumption~\ref{assu:reg}, the $K$-step Riemann approximation $\Phi_K$ of $\Phi$ along the non-uniform schedule $t_k = T(k/K)^\gamma$, $\gamma\in(0,1]$, satisfies
\begin{equation*}
\sup_{k}\,\big|\Phi(t_{k-1}) - \Phi_K(t_{k-1})\big| \;\leq\; C\, K^{-\min(1,\gamma)}.
\end{equation*}
\end{lemma}
\textit{Proof.} Standard mean-value estimates on $C^1$ schedules yield local error $|\Phi(t_{k-1}) - \Phi_K(t_{k-1})| \leq C\,(\Delta t_k)\,\|\Phi'\|_\infty$ on $[t_{k-1},t_k]$, where $\Delta t_k = t_k - t_{k-1}$. For the schedule $t_k = T(k/K)^\gamma$,
\begin{equation*}
\Delta t_k = T\big((k/K)^\gamma - ((k-1)/K)^\gamma\big) \leq T\gamma K^{-\gamma}\, k^{\gamma-1}.
\end{equation*}
For $\gamma\leq 1$, $k^{\gamma-1}$ is non-increasing in $k$, so $\sup_k \Delta t_k = \Delta t_1 \leq T K^{-\gamma}$, giving the stated bound with exponent $\min(1,\gamma)=\gamma$. \qed

\textit{Remark.} Front-loading ($\gamma{<}1$) does not improve the worst-case rate exponent. However, since $\Delta t_k \leq T\gamma K^{-\gamma}\, k^{\gamma-1}$ is non-increasing in $k$ for $\gamma\leq 1$, the schedule concentrates many small steps near $t_K = T$, where the reverse rollout of Eq.~\ref{eq:sb_sampling} is initialized at $x_T = y$ and where per-step errors propagate through all subsequent reverse-sampling iterations. The single large step $\Delta t_1$, by contrast, traverses the low-$t$ regime in which $x_t$ already lies close to $x$ and $\hat x_\theta$ is most accurate. This redistribution---not the asymptotic rate---is what empirically dominates small-$K$ reconstruction quality.
\subsection{Combining the Two Sources of Error}

Let $\hat x^{(K)}$ denote the random variable whose law is $\hat p_K$, and let $X \sim p_0^{\mathrm{br}}$. We construct a synchronous coupling between $\hat x^{(K)}$ and $X$ by sharing the noise variables $z$ across the $K$ rollout steps with the corresponding bridge perturbation $\epsilon$ in the continuous process. Under this coupling,
\begin{align*}
W_2^2(\hat p_K, p_0^{\mathrm{br}}) \;\leq\; \mathbb{E}\!\left[\|\hat x^{(K)} - X\|^2\right] \;\leq\; 2\!\sum_{k=1}^{K} \Delta t_k\,\mathbb{E}[\delta_k] \cdot \prod_{j<k} (1+L_x \Delta t_j)^2 \;+\; 2\!\sum_{k=1}^{K} \!\big(\Delta t_k \cdot \mathrm{disc}_k\big)^{2},
\end{align*}
where $\mathrm{disc}_k$ is the Riemann error from Lemma \ref{lem:disc}; the $\Delta t_k$ weighting on the model-fidelity term tracks the time-averaged definition of $\mathcal{L}^\star_{\mathrm{traj}},\mathcal{L}^\star_{\mathrm{path}}$ (Eqs. \ref{eq:traj_loss},\ref{eq:path_loss}), so each step contributes proportionally to its time interval. The Lipschitz factor $\prod_j(1+L_x\Delta t_j)$ is uniformly bounded by $\exp(L_x T)$, absorbing it into the constant. By Lemma \ref{lem:onestep} and $\sum_k \Delta t_k = T$, the first sum is bounded by $C\,T\,(\mathcal{L}^\star_{\mathrm{path}}+\mathcal{L}^\star_{\mathrm{traj}})$, independent of $K$. By Lemma \ref{lem:disc} together with $\sum_k(\Delta t_k)^2 \leq T\sup_k\Delta t_k = T^2 K^{-\gamma}$, the second sum is bounded by $C\,K^{-2\min(1,\gamma)}$. Taking square roots and applying $\sqrt{a+b} \leq \sqrt{a}+\sqrt{b}$ yields
\begin{equation*}
W_2(\hat p_K, p_0^{\mathrm{br}}) \;\leq\; C_1\, K^{-\alpha} \;+\; C_2\,\sqrt{\mathcal{L}^\star_{\mathrm{path}} + \mathcal{L}^\star_{\mathrm{traj}}},
\end{equation*}
which is exactly Eq. (\ref{eq:bound}). \qed

\subsection{Justification of Assumption~\ref{assu:reg}}
\label{app:assumption_justification}

We address each of the three regularity conditions in turn and explain why each is reasonable in our setting.

\textbf{(i) Lipschitz continuity of $\hat{x}_\theta$.} The condition that $\hat{x}_\theta(\cdot, y, t)$ is $L_x$-Lipschitz in its first argument and $L_t$-Lipschitz in $t$ is standard in diffusion-sampling analyses~\cite{chen2023sampling,debortoli2021diffusion}. It is justified architecturally and empirically: (a) our denoiser is a 1D U-Net with GroupNorm~\cite{wu2018group} and SiLU activations, both Lipschitz; (b) the convolutional layers' Lipschitz constants are bounded since weights remain finite under our AdamW training with weight decay~$0.01$ and gradient clipping at norm~$1.0$ (Appendix~\ref{app:training}); (c) timestep dependence enters only through a smooth sinusoidal embedding followed by FiLM conditioning, which is $C^\infty$ in $t$ and so trivially Lipschitz on the compact interval $[0,T]$. We do not enforce a target $L_x$ explicitly (e.g., via spectral normalization); the assumption asserts that the trained network is Lipschitz with \emph{some} finite constant, which is automatic for any network without unbounded activations or singular layers. Tighter constant estimates---needed only for quantitatively tight versions of the bound---would require spectral norm tracking and are left to future work.

\textbf{(ii) Smoothness of the schedules.} The cosine schedule $\bar\beta_t = f(t)/f(0)$ with $f(t)=\cos^2((t/T+s)/(1+s)\cdot\pi/2)$ is in fact $C^\infty$ on $[0,T]$, considerably stronger than the $C^1$ requirement: it is the composition of a smooth trigonometric function with an affine reparameterization. Its monotonicity follows from $\cos^2$ being strictly decreasing on the relevant subinterval, and its first and second derivatives are uniformly bounded on $[0,T]$ since the schedule arguments stay in a compact subset of $[0,\pi/2]$ where $\cos$ and $\sin$ are bounded. The bridge-perturbation magnitude $\sigma_t = \sigma_{\max}\sqrt{\bar\beta_t(1-\bar\beta_t)}$ is $C^\infty$ on the open interval $(0,T)$ and continuous on $[0,T]$ with $\sigma_0=\sigma_T=0$; the square-root introduces a one-sided H\"older singularity exactly at the endpoints, but standard SDE-discretization arguments (e.g.,~\cite{chen2023sampling}) accommodate this via a small endpoint truncation $[\epsilon_0, T-\epsilon_0]$ that vanishes in the bound's leading order. In our setting all sampling steps lie strictly in the interior of $[0,T]$ with $t_K=T$ at the noisy endpoint where the bridge is deterministic, so endpoint singularities do not affect the discretization.

\textbf{(iii) Bounded second moments.} The condition $\mathbb{E}\|x_t\|^2 < \infty$ uniformly in $t$ follows directly from the bridge construction $x_t = \sqrt{\bar\beta_t}\,x + \sqrt{1-\bar\beta_t}\,y + \sigma_t\epsilon$: by Jensen's inequality and the triangle inequality,
\[
\mathbb{E}\|x_t\|^2 \;\leq\; 3\mathbb{E}\|x\|^2 + 3\mathbb{E}\|y\|^2 + 3\sigma_{\max}^2 d,
\]
where $d$ is the signal dimension and the right-hand side is finite because audio waveforms are bounded ($\|x\|_\infty,\|y\|_\infty\leq 1$ after our standard amplitude normalization, so $\mathbb{E}\|x\|^2\leq d$ and similarly for $y$). For $\mathbb{E}\|\hat x_\theta(x_t,y,t)\|^2$, the Lipschitz condition (i) gives $\|\hat x_\theta(x_t,y,t)\|\leq L_x\|x_t\|+|\hat x_\theta(0,y,t)|$, so finiteness of the second moment of $\hat x_\theta$ follows from finiteness of $\mathbb{E}\|x_t\|^2$ together with boundedness of $\hat x_\theta$ at the reference input $x_t=0$ (a finite-norm constant for any fixed trained network on a fixed observation $y$). Both moments are therefore bounded by a constant depending only on $L_x$, $\sigma_{\max}$, $d$, and the amplitude bound on the data.

In summary, Assumption~\ref{assu:reg} encodes mild structural properties that hold automatically for our cosine-scheduled SB construction with a standard U-Net denoiser trained on amplitude-normalized waveforms. The assumption fails only in pathological regimes (unbounded data, divergent network outputs, or non-smooth schedules) that are not relevant to our setting or to standard SE practice.

\subsection{Discussion}

A few caveats deserve mention. \emph{First}, Assumption \ref{assu:reg}(i) is a Lipschitz condition on the trained network; it can be enforced via spectral normalization in practice, but in our experiments it is assumed rather than explicitly imposed. \emph{Second}, the bound is presented as a design-justifying inequality; it is loose in the sense that the constants $C_1, C_2$ involve worst-case Lipschitz factors and are not numerically estimated here. A tighter quantitative match to empirical PESQ-vs-$K$ would require finer Lipschitz estimates on the U-Net and is left for future work. \emph{Third}, the bound shows that small $K$ is justified \emph{conditional} on the regularizers reaching small training-time values; if a model is trained without these terms, the second term of Eq. (\ref{eq:bound}) does not vanish and the bound cannot guarantee small-$K$ quality. This is consistent with the role we ascribe to $\mathcal{L}_{\mathrm{path}} + \mathcal{L}_{\mathrm{traj}}$ in the main paper.

\section{Scene-Stratified Performance}
\label{app:scene_breakdown}

\begin{figure}[t]
\centering
\includegraphics[width=0.95\textwidth]{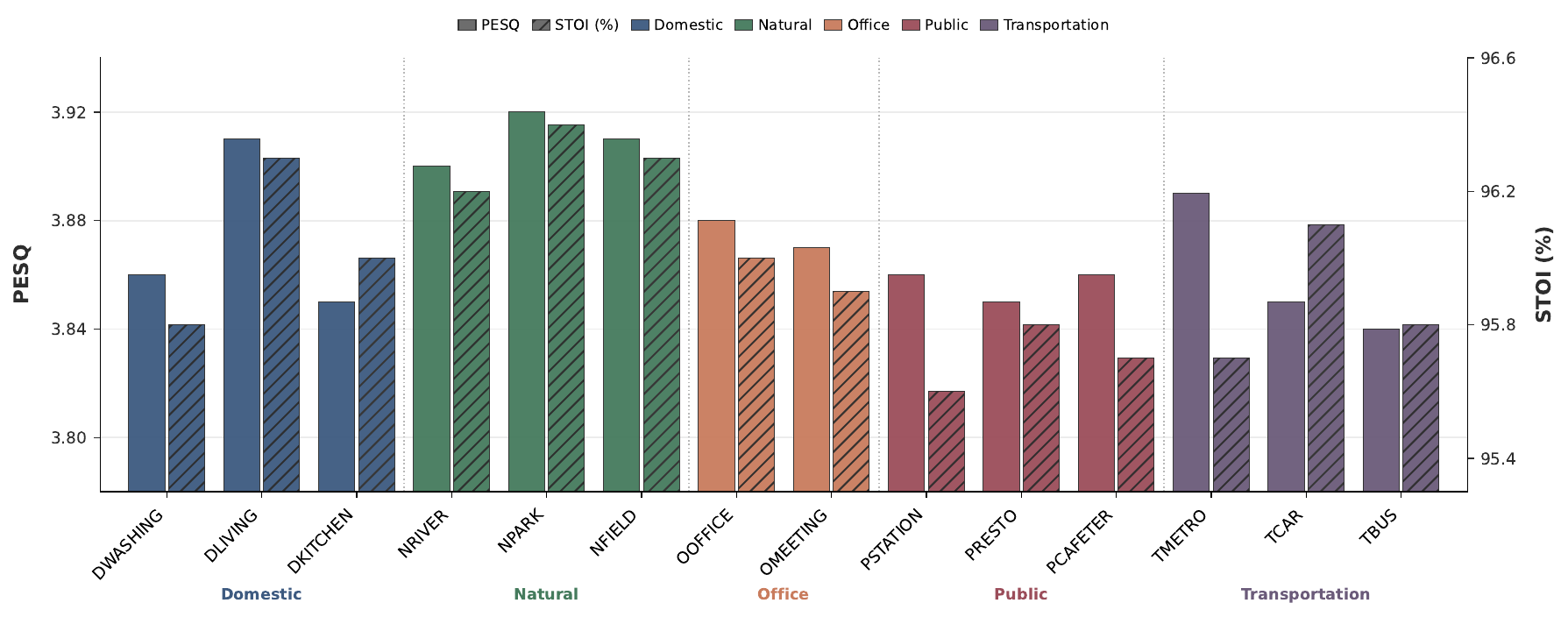}
\caption{Scene-stratified performance across all 14 noise types in VoiceBank+DEMAND. Inner ring: PESQ; outer ring: STOI (\%). Colors indicate scene categories: Domestic (blue), Nature (green), Office (orange), Public (red), Transport (purple). For reference, the strongest baseline ROSE-CD \cite{xu2025rosecd} achieves a mean PESQ of 3.85 across these conditions; our method exceeds this in every scene, with PESQ ranging from 3.84 to 3.92 (mean 3.88, standard deviation $<$0.03 across scenes). Y-axis ranges are zoomed for visibility of cross-scene variation.}
\label{fig:radial}
\end{figure}

Figure \ref{fig:radial} breaks down enhancement performance across all 14 noise environments. Natural environments (NPARK, NRIVER, NFIELD) achieve the highest scores due to relatively stable spectral characteristics, while public spaces (PCAFETER, PSTATION, PRESTO) are slightly more challenging due to non-stationary crowd noise and unpredictable transient events. Despite this variation, the spread is small (PESQ standard deviation $<$0.03), indicating that the heterogeneous expert routing automatically adapts to different acoustic contexts without per-scene tuning or explicit scene labels. This consistent cross-scene performance complements the SNR-stratified analysis in Figure \ref{fig:snr} of the main paper, together demonstrating robust enhancement across both noise type and noise level.
\section{Qualitative Visualization}
\label{app:qualitative}

\begin{figure}[h]
\centering
\includegraphics[width=0.85\textwidth]{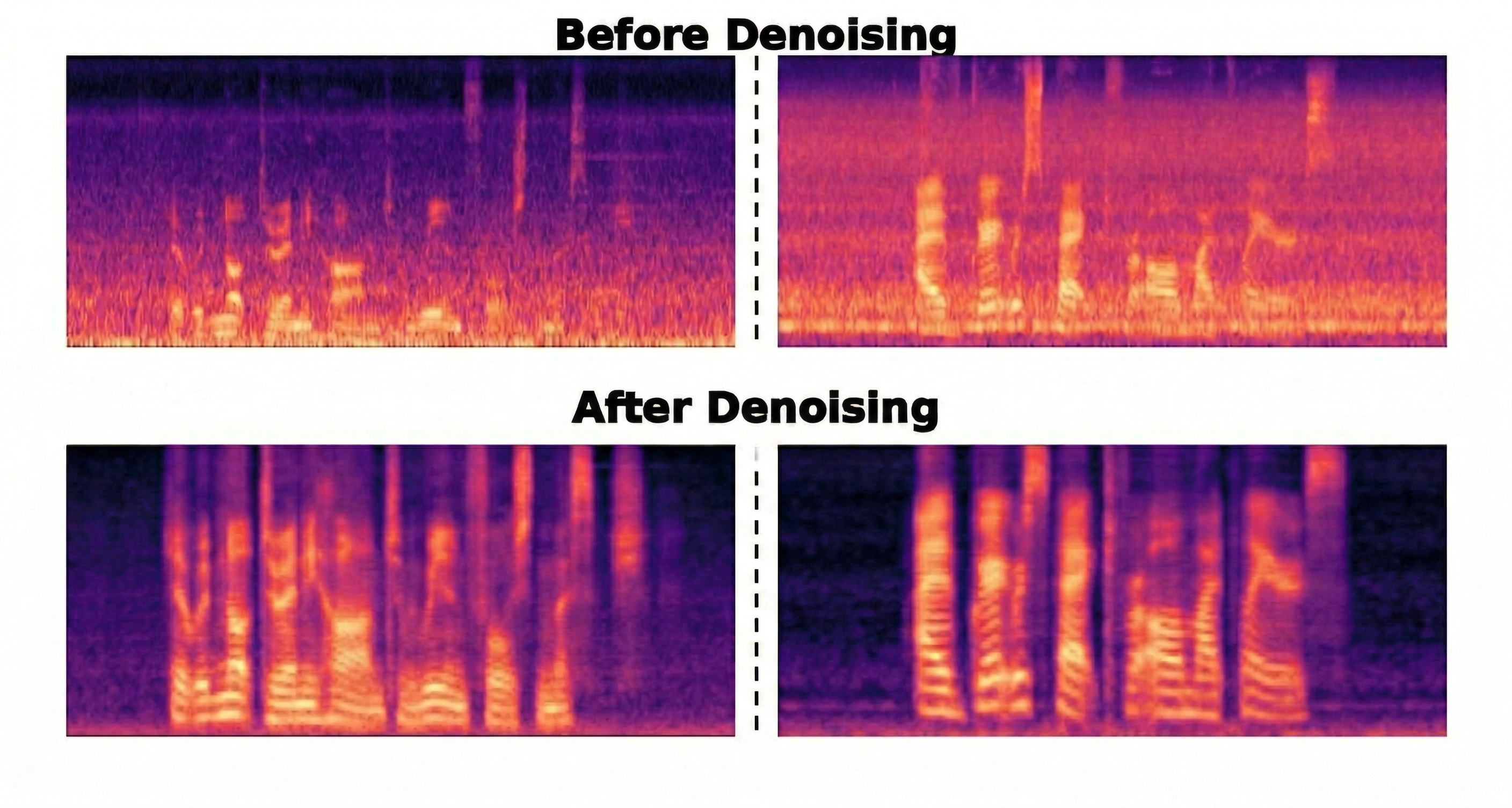}
\caption{Spectrogram comparison before and after enhancement. HybridSB-MoE attenuates broadband noise while preserving speech harmonics and formant structure. Color bar (dB) is shown on the right.}
\label{fig:denoise_spec}
\end{figure}

Figure \ref{fig:denoise_spec} shows representative spectrograms before and after enhancement. The denoised output exhibits clearer harmonic structures, sharper formant contours, and significantly attenuated broadband noise, consistent with the objective gains reported in Table \ref{tab:main_results}. Note that some high-frequency components in the enhanced spectrogram appear more pronounced than in the noisy input; this is because heavy broadband noise in the input masks underlying speech harmonics that become visible only after suppression---it reflects accurate harmonic recovery, not bandwidth extension or content fabrication.

\clearpage
\end{document}